\documentclass[a4paper,11pt]{article}

\usepackage{amsmath,amsfonts,amscd,amssymb, amsthm,mathrsfs, ascmac}
\usepackage{latexsym}
\usepackage{longtable,geometry}
\usepackage[english]{babel}
\usepackage[pdftex]{graphicx}
\usepackage{bbm}
\usepackage{enumitem}
\usepackage{nicefrac}
\usepackage{xcolor}
\usepackage{bm}
\usepackage{here}
\usepackage{authblk}
\usepackage{tikzit}
\tikzstyle{new style 0}=[fill=white, draw=black, shape=circle]
\tikzstyle{new style 1}=[fill=white, draw=red, shape=circle]
\tikzstyle{Box}=[fill={rgb,255: red,191; green,191; blue,191}, draw=black, shape=rectangle, minimum width=0.9cm, minimum height=1.5cm]

\tikzstyle{Red Arrow1}=[->, draw=red]
\tikzstyle{Red Arrow2}=[<-, draw=red]
\tikzstyle{Dashed Line}=[-, draw=black, dashed]
\tikzstyle{RDashed Line}=[-, draw=red]
\tikzstyle{BDashed Line}=[-, draw=blue]
\tikzstyle{Very Thick Line}=[-, draw=black, thick]
\tikzstyle{GVery Thick Line}=[-, draw=green, thick]
\tikzstyle{new edge style 0}=[-]
\tikzstyle{Very Thick Arrow1}=[->, draw=black, very thick]
\tikzstyle{Very Thick Arrow1}=[<-, draw=black, very thick]

\title{\textsf{Thermal stability of the Nagaoka-Thouless theorems
}}
\date{}

\author{Tadahiro Miyao}

\affil{Department of Mathematics, Hokkaido University

Sapporo 060-0810, Japan}

\newcommand{\h}{\mathfrak{H}}

\newcommand{\D}{\mathrm{dom}}
\newcommand{\R}{\mathrm{ran}}

\newcommand{\Fock}{\mathfrak{F}}
\newcommand{\Ffin}{\mathfrak{F}_{\mathrm{fin}}}

\newcommand{\dG}{d\Gamma}

\newcommand{\la}{\langle}
\newcommand{\ra}{\rangle}
\newcommand{\Tr}{\mathrm{Tr}}

\newcommand{\BbbR}{\mathbb{R}}
\newcommand{\BbbN}{\mathbb{N}}
\newcommand{\BbbZ}{\mathbb{Z}}

\newcommand{\BbbC}{\mathbb{C}}

\newcommand{\vepsilon}{\varepsilon}

\newcommand{\Hf}{H_{\mathrm{f}}}

\newcommand{\no}{\nonumber \\}

\newcommand{\Ex}{\mathrm{E}}
\newcommand{\bs}{\boldsymbol}
\newcommand{\Me}{\nu_{\beta}}

\newcommand{\Mei}{\nu_{\beta, \infty} }

\newcommand{\rr}{\mathrm{rad}}
\newcommand{\hh}{\mathrm{HH}}

\newcommand{\bfx}{x}
\newcommand{\bfy}{y}

\newcommand{\bfA}{A}
\newcommand{\bfr}{r}

\newcommand{\rf}{\mathrm{F}}
\newcommand{\rp}{\mathrm{P}}
\newcommand{\rh}{\mathrm{H}}

\begin{document}

\newtheorem{define}{Definition}[section]
\newtheorem{Thm}[define]{Theorem}
\newtheorem{Prop}[define]{Proposition}
\newtheorem{lemm}[define]{Lemma}
\newtheorem{rem}[define]{Remark}
\newtheorem{assum}{Condition}
\newtheorem{example}{Example}
\newtheorem{coro}[define]{Corollary}

\maketitle

\begin{abstract}
We prove that 
the Aizenman-Lieb theorem on ferromagnetism in the Hubbard model holds true even if the electron-phonon
 interactions and the electron-photon interactions are taken into account. Our proof is based on path integral representations of the partition functions.
\begin{flushleft}
{\bf Mathematics Subject Classification (2010).} 
\end{flushleft}
Primary:   82B20, 47D08\\
Secondary: 81T25, 60J28
\begin{flushleft}
\end{flushleft}
 
\end{abstract}

\section{Introduction}\label{Intro}
\setcounter{equation}{0}
\subsection{Background}
The Hubbard model  of  interacting electrons occupies a special place  in the study of ferromagnetism;  this is because it is the simplest model which can  describe  the following fundamental properties:
\begin{itemize}
\item  the Pauli exclusion principle;
\item  the Coulomb repulsion between electrons;
\item  itinerancy of the  electrons.
\end{itemize}
It is believed that ferromagnetism arises from the interplay of these properties.
However, to reveal  the mechanism of ferromagnetism has been mystery,  even today. 
A first rigorous example of the ferromagnetism in the Hubbard model was constructed by Nagaoka and Thouless \cite{Nagaoka, Thouless}. They proved that the ground state of the model exhibits ferromagnetism when 
 there exists precisely one hole
 and the Coulomb strength $U$ is very large.
 At a first glance, it appears that the Nagaoka-Thouless theorem is based on  unrealistic conditions; however, 
  experimental evidence of Nagaoka-Thouless ferromagnetism  was recently presented 
by  using a quantum dot device \cite{DMM}.
The Nagaoka-Thouless theorem is  restricted to the ground states.
It is logical as well as important to ask whether  we can extend the theorem to positive temperatures.
This problem was solved by Aizenman and Lieb  \cite{AL}; in the present paper, we call their result   the Aizenman-Lieb theorem, see Theorem \ref{GeneAZNT} for details.
\medskip

In the real world, the electrons are constantly influenced  by the surrounding environment, e.g.,  
the lattice vibrations, the radiation field and the thermal fluctuations.
Therefore, the following question naturally arises: Are  the Nagaoka-Thouless theorem and related properties of many-electron system stable under the influences from the environment?
This question has been studied by the author,  successfully; in \cite{Miyao7, Miyao5,Miyao,Miyao9}, not only the Nagaoka-Thouless theorem but also  Lieb's theorem\footnote{ This theorem claims  that 
 with a bipartite lattice and a half-filled band, the ground state of the repulsive Hubbard model has total spin $S=||A|-|B||/2$, where $|A|$ (resp. $|B|$) is the
number of sites in the $A$-sublattice (resp. $B$-sublattice), see \cite{Lieb} for details. }
are shown to be stable, even if the electron-phonon and the electron-photon interactions are taken into account; 
furthermore, a general structure behind these stabilities is explored in  \cite{Miyao2}. 
Note that these works are related to  ground state  properties. Can we further take thermal effects into consideration?
The principle purpose of the present paper is to prove stabilities of the Aizenman-Lieb theorem (i.e., an extension of the Nagaoka-Thouless theorem to positive temperatures) under the influences of the lattice vibrations and the quantized radiation field.
\medskip

For later use, we provide  more precise statements of the Nagaoka-Thoulss theorem and the Aizenman-Lieb theorem.
Let us consider a $d$-dimensional square lattice: 
\begin{align}
\Lambda=[-\ell/2, \ell/2)^d\cap \BbbZ^d,
\end{align} 
where $\ell$ is an even number.
The elements of $\Lambda$ are called vertices and we say that $x, y\in \Lambda$ are neighbors if 
$\|x-y\|=1$, where
$\|x\|=\max_{j=1, \dots, d}|x_j|=1$. A pair $e=\{x, y\} \in \Lambda\times \Lambda$ is called edge
 if $x$ and $y$ are neighbors.  We denote by $E_{\Lambda}$ the set of all edges. Clearly, $(\Lambda, E_{\Lambda})$ becomes a graph.
 
The Hubbard model on  $\Lambda$ is defined by the Hamiltonian
\begin{align}
H_{\mathrm{H}, 0}=\sum_{\sigma=\pm 1} \sum_{x, y\in \Lambda} (-t_{xy}) c_{x\sigma}^*c_{y\sigma}
+U \sum_{x\in \Lambda} n_{x, +1} n_{x, -1}+\sum_{{x, y\in \Lambda}\atop{x\neq y}} U_{xy} n_xn_y.
\end{align}
The self-adjoint operator $H_{\mathrm{H}, 0}$ acts in the $N$-electron space $
\mathfrak{H}_{\mathrm{H}}^{(N)}=\bigwedge^N\big(
\ell^2(\Lambda) \oplus \ell^2(\Lambda)
\big)
$, where $\bigwedge^N$ indicates the $N$-fold antisymmetric tensor product.
$c_{x\sigma}^*$ and $c_{x\sigma}$ are the fermionic creation- and annihilation operators satisfying 
the usual anticommutation relations:
\begin{align}
\{c_{x\sigma}^*, c_{y\tau}\}=\delta_{\sigma\tau}\delta_{xy},\ \ \ \{c_{x\sigma}, c_{y\tau}\}=0,
\end{align}
where $\delta_{ab}$ is the Kronecker delta. The number operators are defined by $n_{x\sigma}=c_{x\sigma}^*c_{x\sigma},\ \sigma=\pm 1$
 and $n_x=n_{x, +1}+n_{x, -1}$.
  For simplicity,  the hopping matrix $(t_{xy})$ satisfies the following:
\begin{itemize}
\item[{\bf (T)}] $\displaystyle
t_{xy} =
\begin{cases}
t>0 & \mbox{ if $\{x, y\}\in E_{\Lambda}$}\\
0 & \mbox{otherwise.}
\end{cases}
$
\end{itemize}

 $U$ and $U_{xy}$ are the local and non-local Coulomb matrix elements, respectively.
 In the present paper, we always assume the following:
 \begin{description}
 \item[{\bf (U. 1)}] $U>0$;
 \item[{\bf (U. 2)}]  $U_{xy}=U_{yx}\in \BbbR$ for all $x, y\in \Lambda$ with $x\neq y$.
 \end{description}

 The spin operators at $x\in \Lambda$ are defined by 
 \begin{align}
 S_x^{(j)}=\frac{1}{2}\sum_{\sigma, \sigma'=\pm 1} c_{x\sigma}^*\big(
 s^{(j)}
 \big)_{\sigma\sigma'} c_{x\sigma'},\ \ j=1, 2, 3,
  \end{align} 
  where $s^{(j)}\, (j=1, 2, 3)$ are the $2\times 2$  Pauli matrices.
 The total spin operators are defined by 
 \begin{align}
 S_{\mathrm{tot}}^{(j)}=\sum_{x\in \Lambda}S_x^{(j)},\ \ \ j=1, 2, 3
 \end{align}
 and 
 \begin{align}
 {\bs S}_{\mathrm{tot}}^2=\sum_{j=1}^3\big(
 S_{\mathrm{tot}}^{(j)}
 \big)^2
 \end{align}
 with eigenvalues $S(S+1)$, where $S=0, 1, \dots, N/2$ or $S=1/2, 3/2, \dots, N/2$.
 The Hubbard Hamiltonian   in a uniform magnetic field ${\bs h}=(0, 0, 2b)$ is given by 
 \begin{align}
 H_{\mathrm{H}}=H_{\mathrm{H}, 0}-2bS_{\mathrm{tot}}^{(3)},\ \ b>0.
 \end{align}

 Let us  derive an effective Hamiltonian describing the system with very large $U$.
For this purpose,  we introduce  the {\it Gutzwiller projection}  $P_{\mathrm{G}}$ by 
 \begin{align}
 P_{\mathrm{G}}=\prod_{x\in \Lambda} (1-n_{x, +1} n_{x, -1}).
 \end{align}
 $P_{\mathrm{G}}$ is the orthogonal projection onto the subspace with no doubly occupied sites.

\begin{Prop}[\cite{Miyao}]
 Let us consider the Hubbard model.
 Assume that  $N=|\Lambda|-1$.
We define an effective Hamiltonian $
H_{\mathrm{H}, \infty}
$ by $
H_{\mathrm{H}, \infty}=P_{\mathrm{G}} H_{\mathrm{H}}^{U=0}P_{\mathrm{G}}
$, where $
H_{\mathrm{H}}^{U=0}
$ is the Hubbard Hamiltonian $H_{\mathrm{H}}$ with $U=0$. For each $d\in \BbbN$,  we have
\begin{align}
\lim_{U\to \infty} (H_{\mathrm{H}}-z)^{-1}=(H_{\mathrm{H}, \infty}-z)^{-1}P_{\mathrm{G}},\ \ z\in \BbbC\backslash \BbbR
\end{align}
in the operator norm topology.
\end{Prop}
We denote the restriction of $H_{\mathrm{H}, \infty}$ to $\R(P_{\mathrm{G}})$ by the same symbol.
The Nagaoka-Thouless theorem can be stated as follows.
\begin{Thm}[\cite{Tasaki2,Tasaki22}]\label{NTThm}
Let us consider the Hubbard model.
 Assume that $N=|\Lambda|-1$.
  For each $d\ge 2$, the ground state of $H_{\mathrm{H}, \infty}$ has total spin $S
=(|\Lambda|-1)/2$ and is unique apart from the trivial $(2S+1)$-degeneracy.
\end{Thm}

\begin{rem}\upshape
\begin{itemize}
\item 
By \cite[Corollary 2.2]{KomaTasaki}, we have 
\begin{align}
\lim_{b\to +0} \lim_{|\Lambda|\to \infty} \frac{\big\la S_{\mathrm{tot}}^{(3)}\big\ra_{\mathrm{H}, \infty}(b)}{|\Lambda|}
\ge \sqrt{3} \lim_{|\Lambda|\to \infty} \sqrt{
\frac{
\big\la \big(S_{\mathrm{tot}}^{(3)}\big)^2
\big\ra_{\mathrm{H}, \infty}(b=0)
}{|\Lambda|^2}
}=\frac{1}{2},
\end{align}
where $\la \cdot \ra_{\mathrm{H}, \infty}(b)$ is the ground state expectation associated with $H_{\mathrm{H}, \infty}$:
$
\la O\ra_{\rh, \infty}(b)=\frac{1}{2S+1}\sum_{m=-S}^S \la \psi_m|O\psi_m\ra.
$ Here, $\psi_m$ is the  normalized  unique ground state of $H_{\mathrm{H}, \infty}$ in the $m$-subspace
$\h_{\mathrm{H}, M}^{(N)}[m]=\ker\big(
S_{\mathrm{tot}}^{(3)}-m
\big),\ m\in \mathrm{spec}(S_{\mathrm{tot}}^{(3)})$.
Because $\frac{S_{\mathrm{tot}}^{(3)}}{|\Lambda|} \le \frac{1}{2}$, we arrive at 
\begin{align}
\lim_{b\to +0} \lim_{|\Lambda|\to \infty} \frac{\la S_{\mathrm{tot}}^{(3)}\ra_{\mathrm{H}, \infty}(b)}{|\Lambda|}=\frac{1}{2}.
\end{align}
As we will see, this conclusion is a key to understand extensions of this theorem to positive temperatures.
\item 
The condition $d\ge 2$ is needed in order to guarantee the hole connectivity, see \cite{Tasaki22} for details.
In contrast to this, a weak version of the Nagaoka-Thouless theorem holds true for arbitrary dimension.\footnote{
To be precise, we have the following:
Assume that $N=|\Lambda|-1$.
  For every $d\in \BbbN$, 
  among  the ground states there is  at least  $(2S+1)$ states with  $S
=(|\Lambda|-1)/2$.}
\end{itemize}
\end{rem}

Let 
\begin{align}
Z_{\mathrm{H}, \infty}(\beta)=\Tr_{
\mathfrak{H}_{\mathrm{H}, \infty}^{(N)}} \Big[
e^{-\beta H_{\mathrm{H}, \infty}}
\Big].\label{PartFunctHubbInfty}
\end{align}
We define the thermal expectation values of operators by  
\begin{align}
\la O\ra_{\mathrm{H},  \infty}(b; \beta)=\Tr_{
\mathfrak{H}_{\mathrm{H}, \infty}^{(N)}
}[Oe^{-\beta H_{\mathrm{H}, \infty}}]\Big/Z_{\mathrm{H}, \infty}(\beta).
\end{align}

\begin{Thm}[The Aizenman-Lieb  theorem  \cite{AL}]\label{GeneAZNT}
Let us consider  the Hubbard model. Suppose that $N=|\Lambda|-1$.
For all $d\in \BbbN$,  $0<\beta<\infty$ and $ 0<b$, we obtain 
\begin{align}
\big\la S_{\mathrm{tot}}^{(3)}\big\ra_{\mathrm{H},  \infty}(b; \beta)
>\frac{N}{2}\tanh(\beta b). \label{GNagaH}
\end{align}
\end{Thm}
\begin{rem}\upshape
By \eqref{GNagaH}, we have
\begin{align}
\lim_{b\to +0} \lim_{\beta\to \infty} \lim_{|\Lambda|\to \infty} \frac{\big\la S_{\mathrm{tot}}^{(3)}\big\ra_{\mathrm{H}, \infty}(b; \beta)}{|\Lambda|}
=\lim_{b\to +0}  \lim_{|\Lambda|\to \infty} \lim_{\beta\to \infty} \frac{\big\la S_{\mathrm{tot}}^{(3)}\big\ra_{\mathrm{H}, \infty}(b; \beta)}{|\Lambda|}
=\frac{1}{2}.
\end{align}
 In this sense,
Theorem \ref{GeneAZNT} is an extension of the Nagaoka-Thouless theorem to positive temperatures.
\end{rem}

\subsection{Models}

\subsubsection{The Holstein-Hubbard model}
We consider the interaction between the electrons and the lattice vibrations. The Holstein-Hubbard model is widely accepted as a standard model describing such a  system.
The Hamiltonian of the Holstein-Hubbard model  is given by 
\begin{align}
H_{\hh}=H_{\rh}
+\sum_{x, y\in \Lambda}
g_{xy}n_{x}(b_y^*+b_y)+\sum_{x\in \Lambda} \omega b_x^*b_x. \label{ExtendedHH}
\end{align} 
$H_{\hh}$ acts in the Hilbert space $\mathfrak{H}_{\hh}^{(N)}=\mathfrak{H}_{\rh}^{(N)}\otimes \Fock_{\hh}$, where
$\Fock_{\hh}=\Fock(\ell^2(\Lambda))$,  the bosonic Fock space over $\ell^2(\Lambda)$;
in general, the bosonic Fock space over $\mathfrak{X}$ is 
  defined by 
  \begin{align}
\Fock(\mathfrak{X})=\bigoplus_{n=0}^{\infty} \otimes_{\mathrm{s}}^n \mathfrak{X},
\end{align}
where $\otimes_{\mathrm{s}}^n \mathfrak{X}$ is the $n$-fold symmetric
tensor product of $\mathfrak{X}$ with $\otimes_{\mathrm{s}}^0\mathfrak{X}=\BbbC$.
$b_x^*$ and $b_x$ are the bosonic creation- and annihilation operators
at site $x$ satisfying the standard commutation relations:
\begin{align}
[b_x, b_y^*]=\delta_{xy},\ \ \ [b_x, b_y]=0
\end{align}
on $\Fock_{\hh, \mathrm{fin}}=\Ffin(\ell^2(\Lambda))$,  where $\Fock_{\mathrm{fin}}(\mathfrak{X})$ is the indirect direct sum of $\otimes_{\mathrm{s}}^n\mathfrak{X}$.
$g_{xy}$ is the strength of the electron-phonon interaction.
The phonons are assumed to be dispersionless with energy $\omega>0$.
Henceforth, we assume the following:
\begin{flushleft}
{\bf(G)} $(g_{xy})_{x, y}$ is a real symmetric matrix.
\end{flushleft} 
Using the Kato-Rellich theorem \cite[Theorem X.12]{ReSi2},  one can prove  that  $H_{\hh}$ is self-adjoint on $\D(\sum_{x\in \Lambda} b_x^*b_x)$  and  bounded from below, where, for linear operator $A$,  $\D(A)$ indicates the domain of $A$.

\subsubsection{The Hubbard model coupled to the quantized radiation field}
We consider an  $N$-electron system coupled to the quantized radiation
field. 
Suppose that 
the lattice $\Lambda$  is embedded into the region $V=[-L/2,
L/2]^3\subset \BbbR^3$
with $L>0$. Thus, when we consider this system, we assume that $d\le 3$.\footnote{
We can  consider   the systems with $d\ge 4$ by extending the definitions of operators associated with the quantized radiation fields to higher dimensions. However,   for simplicity, we restrict our attention to the  case where $d\le 3$.
}
The system is described by the following   Hamiltonian
\begin{align}
H_{\rr}=&\sum_{ x, y\in \Lambda}\sum_{\sigma=\pm 1}
(-t_{xy})
\exp\Bigg\{
i   \int_{C_{\bfx\bfy}} d\bfr\cdot  \bfA(\bfr)\Bigg\}c_{\bfx\sigma}^*
 c_{\bfy\sigma}
+\sum_{k\in V^*}\sum_{\lambda=1,2} \omega(k) a(k, \lambda)^*a(k, \lambda)\no
&+ U\sum_{x\in \Lambda} n_{x, +1} n_{x, -1}+
\sum_{{x, y\in
 \Lambda}\atop{ x\neq y}}U_{xy}n_{x}n_y.
\label{DefHamiltonian}
\end{align}
$H_{\rr}$ acts in the Hilbert space $
\mathfrak{H}_{\rr}^{(N)}=\mathfrak{H}_{\rh}^{(N)}
 \otimes
\Fock_{\rr}$.
$\Fock_{\rr}$ is the Fock space over $\ell^2(V^*\times \{1, 2\})$ with
$V^*=(\frac{2\pi}{L}\BbbZ)^3$.
$a(k, \lambda)^*$ and $a(k, \lambda)$
 are the bosonic creation- and annihilation operators, respectively. These
 operators satisfy the following commutation relations: 
 \begin{align}
[a(k, \lambda), a(k', \lambda')^*]=\delta_{\lambda\lambda'}
  \delta_{kk'},\ \ 
[a(k, \lambda), a(k', \lambda')]=0
\end{align} 
on $\Fock_{\rr, \mathrm{fin}}:=\Ffin(\ell^2(V^*\times \{1, 2\}))$.
The quantized vector potential is given by 
\begin{align}
A(x)=|V|^{-1/2}\sum_{k\in V^*}
 \sum_{\lambda=1,2}\frac{\chi_{\kappa}(k)}{\sqrt{2\omega(k)}}\vepsilon(k,
 \lambda) \Big(
e^{ik\cdot x}a(k, \lambda)+e^{-ik\cdot x}a(k, \lambda)^*
\Big).
\end{align} 
The form factor $\chi_{\kappa}$ is the indicator function of the ball of
radius $0<\kappa<\infty$, centered at the origin.
The dispersion relation $\omega(k)$ is chosen to be $\omega(k)=|k|$ for
$k\in V^* \backslash \{0\}$, $\omega(0)=m_0$ with $0<m_0<\infty$.
For concreteness, the polarization vectors are chosen as  
\begin{align}
\vepsilon(k, 1)
=\frac{(k_2, -k_1, 0)}{\sqrt{k_1^2+k_2^2}}
,\ \  \vepsilon(k, 2)=\frac{k}{|k|} \wedge \vepsilon(k, 1).
\end{align} 
To avoid ambiguity, we set $
\vepsilon_j(k, \lambda)=1/\sqrt{3},\ j=1, 2, 3
$ if $k_1=k_2=0$.\footnote{
We can generalize this condition to $\vepsilon(k, \lambda)=a_{\lambda}\in \BbbR^3$ with $|a_{\lambda}|=1$ for $k_1=k_2=0$.
We can also choose $\vepsilon(k, \lambda)=0$, if $k_1=k_2=0$. Many of  arguments below are unaffected by the  choice   of $\vepsilon(k, \lambda)$.
}
 $A(x)$ is essentially self-adjoint. We denote its
closure by the same symbol. $C_{xy}$ is a piecewise smooth curve from $x$ to $y$.
Note that the more precise definition of $\int_{C_{xy}} A(r)\cdot dr$ will be given in Section \ref{TraceBoseGene}.
This model was introduced by Giuliani {\it et al.} in \cite{GMP, GMP2}.
 $H_{\rr}$ is essentially self-adjoint and bounded from below.
We denote its closure by the same symbol.

\subsection{Results}

We will display an extension of Theorem \ref{GeneAZNT}. For this purpose, we need the following proposition.

\begin{Prop}[\cite{Miyao}]
Assume that  $N=|\Lambda|-1$. 
For $\natural =\rr, \hh$,
we define an effective Hamiltonian $
H_{ \natural, \infty}
$ by $
H_{\natural , \infty}=P_{\mathrm{G}} H_{\natural }^{U=0}P_{\mathrm{G}}
$, where $
H_{\natural }^{ U=0}
$ is the corresponding  Hamiltonian $H_{\natural }$ with $U=0$. Then we have
\begin{align}
\lim_{U\to \infty} (H_{\natural }-z)^{-1}=(H_{\natural,  \infty}-z)^{-1}P_{\mathrm{G}},\ \ z\in \BbbC\backslash \BbbR
\end{align}
in the operator norm topology.
\end{Prop}
We denote the restriction of $H_{\natural, \infty}$ to $\R(P_{\mathrm{G}})$ by the same symbol.
The following theorem is a generalized version of the Nagaoka-Thouless theorem.
\begin{Thm}[\cite{Miyao}]\label{NTThm2}
We have the following:
\begin{itemize}
\item[\rm (i)] 
For $d=2, 3$, 
the ground state of $H_{\mathrm{rad}, \infty}$ has total spin $S
=(|\Lambda|-1)/2$ and is unique apart from the trivial $(2S+1)$-degeneracy.
\item[\rm (ii)]
 For each $d\ge 2$, the ground state of $H_{\mathrm{HH}, \infty}$ has total spin $S
=(|\Lambda|-1)/2$ and is unique apart from the trivial $(2S+1)$-degeneracy.
\end{itemize}
\end{Thm}
\begin{rem}
{\rm 
As before, we have the following:
\begin{align}
\lim_{b\to +0}\lim_{|\Lambda|\to \infty} \frac{\la S_{\mathrm{tot}}^{(3)}\ra_{\natural, \infty}(b)}{|\Lambda|}=\frac{1}{2},\ \ \natural=\rr, \hh,
\end{align}
where $\la \cdot \ra_{\natural, \infty}(b)$ is the ground state expectation associated with $H_{\natural, \infty}$:
$
\la O\ra_{\natural, \infty}(b)=\frac{1}{2S+1}\sum_{m=-S}^S \la \psi_m|O\psi_m\ra
$;  here, $\psi_m$ is the  normalized  unique ground state of $H_{\natural, \infty}$ in the $m$-subspace
$\h_{\natural, M}^{(N)}[m]=\h_{\rh, M}^{(N)}[m]\otimes \Fock_{\natural},\ \natural=\rr, \hh$. We also remark  that  there are some other extensions of the Nagaoka-Thouless theorem,  see, e.g., \cite{KT,KSV}.
}
\end{rem}

Theorem \ref{NTThm2} can be extended to positive temperatures as follows.
Let 
\begin{align}
Z_{\natural , \infty}(\beta)=\Tr_{
\mathfrak{H}_{\natural , \infty}^{(N)}} \Big[
e^{-\beta H_{\natural , \infty}}
\Big],\ \ \natural =\rr, \hh. \label{PartFunctNatuInf}
\end{align}
The thermal expectation values of operators are defined by 
\begin{align}
\la O\ra_{\natural , \infty}(\beta; b)=\Tr_{\mathfrak{H}_{\natural , \infty}^{(N)}}\Big[Oe^{-\beta H_{\natural , \infty}}\Big]\Big/Z_{\natural , \infty}(\beta).
\end{align}

The main result in this paper is the following.

\begin{Thm}\label{FinNTEnv}
Suppose that $N=|\Lambda|-1$. Suppose that $d\le 3$ for $\natural=\mathrm{rad}$,  $d\in \BbbN$ for $\natural=\mathrm{HH}$.
For all $0<\beta<\infty$ and $0<b$, we obtain 
\begin{align}
\big\la S_{\mathrm{tot}}^{(3)}\big\ra_{\natural , \infty}(\beta; b)>\frac{N}{2}\tanh(\beta b), \ \  \natural =\mathrm{rad}, \mathrm{HH}. \label{GNaga}
\end{align}

\end{Thm}
We will provide a proof of Theorem \ref{FinNTEnv} in Section \ref{SecPf2}.

\begin{rem}
{\rm 
\begin{itemize}
\item By \eqref{GNaga}, we have
\begin{align}
\lim_{b\to +0} \lim_{\beta\to \infty} \lim_{|\Lambda|\to \infty} \frac{\big\la S_{\mathrm{tot}}^{(3)}\big\ra_{\natural , \infty}(\beta; b)}{|\Lambda|}
=\lim_{b\to +0}  \lim_{|\Lambda|\to \infty} \lim_{\beta\to \infty} \frac{\big\la S_{\mathrm{tot}}^{(3)}\big\ra_{\natural , \infty}(\beta; b)}{|\Lambda|}
=\frac{1}{2}.
\end{align}
In this sense,
Theorem \ref{FinNTEnv} can be regarded as  an extension of Theorem \ref{NTThm2} to positive temperatures.
\item Theorem \ref{FinNTEnv} can be extended to models on a general bipartite lattice with nearest neighbour hopping.
\item Let us consider the multi-polaron model in a bounded region $[-L, L]^N$
with periodic boundary conditions, and a hard-core repulsion.
 (As for the definition of this model, see, e.g., \cite{FLST}.) 
If $N$ is odd, then our method can be applicable to the model. Similar observations hold true for 
the Pauli-Fierz model in $[-L, L]^N$. For the precise definition of this model, we refer to \cite{BFS, GLL,LMS,Spohn}.
\end{itemize}

}
\end{rem}

\subsection{Origanization}
The organization of the present paper is as follows.
In Section \ref{FKIEl}, we construct Feynman-Kac-It\^o formulas for the  magnetic
Hubbard model. In Section \ref{TraceBoseGene}, we provide  trace formulas for free Bose fields in terms of path integral
 representations. By combining the formulas in Sections \ref{FKIEl} and \ref{TraceBoseGene}, we 
 construct Feynman-Kac-It\^o formulas for the partition functions for  $H_{\hh}$ and $H_{\rr}$  in Sections \ref{FKIRR} and \ref{ConstFHH}.
 Section \ref{ConstLP} is devoted to give random loop representations for  the partition functions. Applying these  representations, we give a proof of Theorem \ref{FinNTEnv} in Section \ref{SecPf2}.
 In Section \ref{UseForU}, we derive an interesting formula for the partition functions.
Appendix \ref{AppAUse} is devoted to prove a useful proposition.

\subsection*{ Acknowledgments}
This work was  supported by   KAKENHI 18K0331508.  I am grateful to an anonymous referee for helpful advice.

\section{Feynman-Kac-It\^o formulas for the Hubbard models}\label{FKIEl}
\setcounter{equation}{0}

\subsection{A Feynman-Kac-It\^o formuals for  one-electron Hamiltonians} \label{SingleElFKF}
\subsubsection{Discrete magnetic Schr\"{o}dinger operators}

Let us consider a single electron living in $\Lambda$. One is given a hopping matrix $(t_{xy})$ satisfying {\bf (T)} in Section \ref{Intro}.
The kinetic energy of the electron is a self-adjoint operator $h_0$ acting in $\ell^2(\Lambda) \oplus \ell^2(\Lambda) \cong \ell^2(\Omega)$ with $\Omega=\Lambda\times \{-1, +1\}$  given by
\begin{align}
(h_0f)(x, \sigma)=\sum_{\sigma=\pm 1} \sum_{ y\in \Lambda} t_{xy} \Big(
f(x, \sigma)-f(y, \sigma)
\Big), \ \ \ f\in \ell^2(\Omega). \label{OneElDef}
\end{align}
Note that the inner product of $\ell^2(\Omega)$ is given by $\la f|g\ra_{\ell^2(\Omega)}=\sum_{\sigma=\pm 1} \sum_{x\in \Lambda} f(x, \sigma)^*g(x, \sigma)=\sum_{X\in \Omega} f(X)^*g(X)$.
By a {\it magnetic potential} on $\Lambda$, we understand  a  matrix $\alpha=(\alpha_{xy})$ such that 
\begin{itemize}
\item $\alpha_{xy}\in \BbbR$ for all $\{x, y\}\in E_{\Lambda}$;
\item $\alpha_{xy}=-\alpha_{yx}$ for all $\{x, y\}\in E_{\Lambda}$.
\end{itemize}
The kinetic energy of the electron in the magnetic potential is given by 
\begin{align}
(h_0(\alpha)f)(x, \sigma)=\sum_{\sigma=\pm 1} \sum_{ y\in \Lambda} t_{xy} \Big(
f(x, \sigma)-e^{i\alpha_{xy}}f(y, \sigma)
\Big), \ \ \ f\in \ell^2(\Omega). \label{OneElDef22}
\end{align}
Let $v$ be a potential, i.e., a multiplication operator by the real-valued  function $v$:
$(vf)(x, \sigma)=v(x)f(x, \sigma)$ for all $f\in \ell^2(\Omega)$.
Then {\it discrete magnetic Schr\"{o}dinger  operators} are  defined by 
\begin{align}
h_v(\alpha)=h_0(\alpha)+v.
 \end{align}
Trivially, $h_{v}(\alpha)$ is self-adjoint.

\subsubsection{A Feynman-Kac-It\^o formula for $h_v(\alpha)$}\label{EleAuxFKF}
As a first step, we will construct a Feynman-Kac-It\^o formula for $h_v(\alpha)$.
In this study, we employ a useful   description by G\"{u}neysu, Keller and Schmidt \cite{GKS}.

For notational simplicity, we set $\BbbN_0=\{0\} \cup \BbbN$.
Let $(Y_n)_{n\in  \BbbN_0}$ be a  discrete time Markov chain with state-space $\Omega$,
 which satisfies, for $X=(x, \sigma), Y=(y, \tau)\in \Omega$, 
 \begin{align}
 P(Y_n=X|Y_{n-1}=Y)=
  \delta_{\sigma\tau} \frac{t_{xy}}{d(y)} 
 ,\ \ n\in \BbbN,
 \end{align}
 where  $d(x)=\sum_{y\in \Lambda} t_{xy}$.
Throughout, we work with  fixed probability spaces $(M, \mathcal{F}, P)$.
Let $(T_n)_{n\in \BbbN}$ be independent exponentially  distributed  random variables of parameter $1$, independent of 
$(Y_n)_{n\in \BbbN_0}$. 
Set \begin{align}
S_n=\frac{T_n}{d(Y_{n-1})},\ \ \  J_n=S_1+\cdots +S_n
\end{align}
and 
\begin{align}
X_t=Y_n\ \ \ \mbox{if $J_n\le t< J_{n+1}$ for some $n$}.
\end{align}
$(X_t)_{t\ge 0}$ becomes  a right continuous process. Furthermore, 
$J_0:=0, J_1, J_2, \dots$ are the jump times of $(X_t)_{t\ge 0}$ and 
$S_1, S_2, \dots$ are the holding times of $(X_t)_{t\ge 0}$.
Let $P_X(\cdot)=P(\cdot | X_0=X)$ and let $(\mathcal{F}_t)_{t\ge 0}$ be the filtration defined by 
$\mathcal{F}_t=\sigma(X_s| s\le t)$. Then 
$
(M, \mathcal{F}, (\mathcal{F}_{t})_{t\ge 0}, (P_X)_{X\in \Omega})
$
 is a strong Markov process, see,  e.g., \cite[Theorems 2.8.1 and 6.5.4]{Norris}.

 Set $N(t):=\sup\{n\in \BbbN_0\, |\, J_n\le t\}$, the number of jumps of $(X_t)_{t\ge 0}$ until $t\ge 0$.
For  $s<t$, we introduce  random variables by 
\begin{align}
\int_s^t \alpha(dX_u)&=\sum_{n=N(s)+1}^{N(t)} \alpha_{X_{J_{n-1}}X_{J_n}},\\
\mathcal{S}_{[s, t]}(v, \alpha|X_{\bullet})&=i\int_s^t \alpha(dX_u)-\int_s^t v(X_u) du.
\end{align}
In the above definition, we understand that $\int_s^t \alpha(dX_u)=0$, provided that $N(s)=N(t)$.

Let $g$ be a function on $\Omega$. The multiplication operator associated with $g$ is defined by 
$(M_g f)(X)=g(X)f(X)$ for $f\in \ell^2(\Omega)$.
In what follows, $M_g$ is abbreviated as $g$ for notational simplicity. Trivially, we have $
\|gf\|_{\ell^2(\Omega)}\le \|g\|_{\ell^{\infty}(\Omega)} \|f\|_{\ell^2(\Omega)}
$, where $\|g\|_{\ell^{\infty}(\Omega)}=\max_{X\in \Omega}|g(X)|$.
We denote by $\ell^{\infty}(\Omega)$ the abelian $\mathrm{C}^*$-algebra of multiplication operators on $\ell^2(\Omega)$, equipped with the norm $\|\cdot\|_{\ell^{\infty}(\Omega)}$.

\begin{Prop}\label{SingleEFK}
Let $\alpha_1, \dots, \alpha_n $ be magnetic potentials.
Let $f_0, f_1, \dots, f_n\in \ell^{\infty}(\Omega)$ and $f\in \ell^2(\Omega)$. For each $0<t_1<t_2<\cdots <t_n$ and 
$X\in \Omega$, we have
\begin{align}
&\Big(
f_0 e^{-t_1h_v(\alpha_1)} f_1 e^{-(t_2-t_1)h_v(\alpha_2)}f_2 \cdots f_{n-1}e^{-(t_n-t_{n-1}) h_v(\alpha_n)} f
\Big)(X)\no
=&\Ex_X\Bigg[
f_0(X_0) f_1(X_{t_1}) \cdots f_{n-1}(X_{t_{n-1}}) f(X_{t_n})\exp\Bigg\{\sum_{\ell=1}^n \mathcal{S}_{[t_{\ell-1}, t_{\ell}]}(v, \alpha_{\ell}|X_{\bullet})\Bigg\}
\Bigg],
\end{align}
where $\Ex_X[F]=\int_M dP_X F$ and $t_0=0$.
\end{Prop}
\begin{proof} In \cite[Theorem 4.1]{GKS},  the following Feynman-Kac-It{\^o} formula 
has been established: 
\begin{align}
\big(
e^{-t h_v(\alpha)} f
\big)(X)=\Ex_X\Big[
e^{\mathcal{S}_{[0, t]}(v, \alpha|X_{\bullet})} f(X_t)
\Big], \ \ f\in \ell^2(\Omega),\ \ X\in \Omega.
\end{align}
By this formula and the Markov property of $(X_t)_{t\ge 0}$, we can prove the assertion in Proposition \ref{SingleEFK}. 
\end{proof}

\subsection{A Feynma-Kac-It\^o formula for $N$-electron Hamiltonians}

Let us consider an $N$-electron system. 
The non-interacting Hamiltonian is given by 
\begin{align}
L( \alpha)=\underbrace{h_{v}(\alpha) \otimes 1 \otimes \cdots \otimes 1}_N +1 \otimes h_v(\alpha)\otimes 1 \otimes \cdots \otimes 1 +\cdots+ 1 \otimes \cdots \otimes 1\otimes h_v(\alpha).
\label{NELHAMI}
\end{align}
Note that 
$L(\alpha)$ acts in  $\bigotimes_{j=1}^N\ell^2(\Omega)\cong \ell^2(\Omega^N)$.
In order to take the Fermi-Dirac statistics into consideration, we introduce the {\it antisymmetrizer} $A_N$
on $\ell^2(\Omega^N)$ by 
\begin{align}
(A_NF)({\bs X})=\sum_{\tau\in \mathfrak{S}_N} \frac{\mathrm{sgn}(\tau) }{N!} F(\tau^{-1}{\bs X}),\   F\in \ell^2(\Omega^N), \  {\bs X}=(X^{(1)}, \dots, X^{(N)}) \in \Omega^N, 
\end{align}
where $\mathfrak{S}_N$ is the permutation group on $\{1, \dots, N\}$, and $\tau {\bs X}:=(X^{(\tau(1))}, \dots, X^{(\tau(N))}),\ \tau\in \mathfrak{S}_N$.
As is well-known, $A_N$ is an   orthogonal projection from $\ell^2(\Omega^N)$ onto $\ell^2_{\mathrm{as}}(\Omega^N)$, the set of all antisymmetric functions on $\Omega^N$.
We define $L_{\mathrm{as}}(\alpha)=A_NL(\alpha)A_N$. Note that we can naturally identify $L_{\mathrm{as}}(\alpha)$ with $L(\alpha) \restriction \ell^2_{\mathrm{as}}(\Omega^N)$.

We wish to construct  a Feynman-Kac-It\^o formula for $L_{\mathrm{as}}(\alpha)$. For this purpose, let 
\begin{align}
\Omega_{\neq}^N=\Big\{{\bs X}\in \Omega^N\, \Big|\, X^{(i)} \neq X^{(j)} \ \mbox{for all $i, j\in \{1, \dots, N\}$ with $i\neq j$}\Big\}.
\end{align}
We introduce  an event  by
\begin{align}
D&=D_{{\rm O}}\cap D_{\mathrm{S}}
\end{align}
with
\begin{align}
D_{{\rm O}}&=
\Big\{{\bs m} \in (M)^N\, \Big|\, 
\mbox{
${\bs X}_s({\bs m}) \in \Omega_{\neq}^N$ for all $s\in [0, \infty)$
}
\Big\} ,\\
D_{{\rm S }}&=\Big\{{\bs m} \in (M)^N\, \Big|\, 
\mbox{
$\sigma^{(j)}_s({\bs m})=\sigma^{(j)}_0({\bs m})$ for all $j=\{1. \dots, n\}$ and $s\in [0, \infty)$
}
\Big\},
\end{align}
where in the definition of $D_{\mathrm{S}}$, we  used the following notation: 
\begin{align}
X_s^{(j)}({\bs m})=( x_s^{(j)}({\bs m}),\sigma_s^{(j)}({\bs m})). \label{ComNot}
\end{align}

For each ${\bs m} \in D$, a right-continuous $\Omega^N$-valued function 
\begin{align}
({\bs X}_t({\bs m}))_{t\ge 0}=\big(
X^{(1)}_t({\bs m}), \dots, X_t^{(N)}({\bs m})
\big)
\end{align}
 is simply called a {\it path}; the path $
({\bs X}_t({\bs m}))_{t\ge 0}
$
 represents a trajectory of the $N$-electrons.  
  Let us write $X_t^{(j)}({\bs m})=(x_t^{(j)}({\bs m}), \sigma_t^{(j)}({\bs m}))$; $\sigma^{(j)}_t({\bs m})$ is  called the {\it spin component} of $X_t^{(j)}({\bs m})$, and 
 $
 x_t^{(j)}({\bs m})
 $
  is called the {\it spatial component} of $
  X_t^{(j)}({\bs m})
  $, respectively.
 We note that,  for  ${\bs m}\in D$,   
 the path $({\bs X}_t({\bs m}))_{t\ge 0}$
 possesses  properties such that  
 $
 \sigma^{(j)}_t({\bs m})
 $ are constant in time; and there are no encounters of electrons of equal spin.

Let $F$ be a function on $\Omega^N$. We say that $F$ is {\it symmetric} if 
$F(\tau {\bs X})=F({\bs X})$ for all $\tau\in \mathfrak{S}_N$ and ${\bs X}\in \Omega^N$.
We denote by $\ell^{\infty}_{\mathrm{s}}(\Omega^N)$ the abelian ${\rm C}^*$-algebra of multiplication operators by symmetric functions on $\Omega^N$, equipped with the norm $
\|\cdot\|_{\ell^{\infty}(\Omega^N)}
$.

\begin{Prop}\label{NFermi}
 For every $d\in \BbbN$, 
$F_0, F_1, \dots, F_{n-1}\in \ell^{\infty}_{\mathrm{s}}(\Omega^N),\ F\in \ell^2_{\mathrm{as}}(\Omega^N)$,  $0<t_1<t_2<\cdots <t_n$ and  $
 {\bs X}=\big(X^{(1)},\dots, X^{(N)}\big)\in \Omega_{\neq}^N
$,  we have
\begin{align}
&\Big(F_0 e^{-t_1L_{\mathrm{as}}(\alpha_1)} F_1 e^{-(t_2-t_1)L_{\mathrm{as}}(\alpha_2)}F_2 \cdots F_{n-1}e^{-(t_n-t_{n-1}) L_{\mathrm{as}}(\alpha_n)} F
\Big)({\bs X})\no
=&\Ex_{\bs X}\Bigg[ 1_{D}
F_0({\bs X}_0) F_1({\bs X}_{t_1}) \cdots F_{n-1}({\bs X}_{t_{n-1}}) F({\bs X}_{t_n})\exp\Bigg\{ \sum_{j=1}^N\sum_{\ell=1}^n \mathcal{S}_{[t_{\ell-1}, t_{\ell}]}(v, \alpha_{\ell}\, |  X^{(j)}_{\bullet})\Bigg\}
\Bigg],  \label{LFK}
\end{align}
 where $\Ex_{\bs X}[F]=\int_{\Omega^N} \prod_{j=1}^NdP_{X^{(j)}} F$ and $1_{D}$ is the indicator function of the event $D$.
\end{Prop}
\begin{proof}
It is not so difficult to show  (\ref{LFK}) without the term $1_{D}$ from Proposition \ref{SingleEFK}.
Below, we will explain why the term $1_{D}$ appears in the formula.

  Let $P$ be the multiplication operator by $1_{\Omega^N_{\neq}}$.
We readily confirm that  $P$ is the  orthogonal projection  from $\ell_{\mathrm{as}}^2(\Omega^N)$ onto $\ell^2_{\mathrm{as}}(\Omega_{\neq}^N)$. In addition, it holds that 
\begin{align}
PL_{\mathrm{as}}(\alpha)=L_{\mathrm{as}}(\alpha)P. \label{Commutative}
\end{align}
We denote by ${\bs \sigma}_0=(\sigma^{(1)}_0, \dots, \sigma^{(N)}_0)\in \{-1, +1\}^N$ the set of  spin components of ${\bs X}=(X^{(1)}, \dots, X^{(N)})$
 with $X^{(j)}=(x_0^{(j)}, \sigma^{(j)}_0)$.
Let $Q$ be the multiplication operator by $1_{\{\bs \sigma_0\}}$:
$(QF)({\bs Y})=\delta_{{\bs \sigma_0} {\bs \sigma(Y)}}F({\bs Y}), F\in \ell^2_{\mathrm{as}}(\Omega^N)$, where
${\bs \sigma}({\bs Y})$ is the set of  spin components of ${\bs Y}\in \Omega^N$.
Note that 
\begin{align}
QL_{\mathrm{as}}(\alpha)=L_{\mathrm{as}}(\alpha)Q. \label{Commutative2}
\end{align}

We set $L_0=L_{\mathrm{as}, v=0}(\alpha=0)$. By (\ref{OneElDef}), we know that 
\begin{align}
L_0 1_{\{{\bs \sigma}_0\}}=0. \label{Kernel}
\end{align}
 For $t\ge 0$,  let $D(t)=\big\{{\bs m} \in (M)^N\, \big|\, 
\mbox{
${\bs X}_t({\bs m}) \in \Omega_{\neq}^N$  and ${\bs \sigma}_t({\bs m})={\bs \sigma}_0$
}
\big\}$, where,  as before, ${\bs \sigma}_t({\bs m})$ is the set of  spin components  of $
{\bs X}_t({\bs m})
$.
By setting $1_{\Omega_{\neq}^N \cap \{{\bs \sigma_0}\}}:=1_{\Omega_{\neq}^N} \times 1_{\{{\bs \sigma_0}\}}$, we have, for $0<t_1<t_2<\cdots <t_n$,
\begin{align}
P_{\bs X}\Bigg(
\bigcap_{i=1}^n D(t_i)
\Bigg)
&= \Ex_{\bs X} \Big[
1_{\Omega_{\neq}^N \cap \{{\bs \sigma_0}\}}({\bs X}_{t_1}) 1_{\Omega_{\neq}^N \cap \{{\bs \sigma_0}\}}({\bs X}_{t_2}) \cdots 1_{\Omega_{\neq}^N  \cap \{{\bs \sigma_0}\}}({\bs X}_{t_n})
\Big]\no
&= \Big(
e^{-t_1L_0} 1_{\Omega_{\neq}^N \cap \{{\bs \sigma_0}\}} e^{-(t_2-t_1)L_0} 1_{\Omega_{\neq}^N \cap \{{\bs \sigma_0}\}}\cdots e^{-(t_n-t_{n-1})L_0}1_{\Omega_{\neq}^N \cap \{{\bs \sigma_0}\}}
\Big)({\bs X})\no
&= \Big(e^{-t_nL_0}1_{\Omega_{\neq}^N \cap \{{\bs \sigma_0}\}}\Big)({\bs X})\no
&=1-\Big(e^{-t_nL_0} 1_{(\Omega_{\neq}^N)^c \cap \{{\bs \sigma}_0\}}\Big)({\bs X}),
\end{align}
where $(\Omega_{\neq}^N)^c$ is the complement of $\Omega_{\neq}^N$.
In the third equality, we have used  (\ref{Commutative}) and (\ref{Commutative2}); in the last equality, we have used (\ref{Kernel}).
By using the fact \begin{align}
\lim_{t\to \infty} e^{-tL_0}
1_{(\Omega_{\neq}^N)^c \cap \{{\bs \sigma}_0\}}
=0,
\end{align}
 we have 
$
P_{\bs X}\Big(\bigcap_{t\in \mathbb{Q}}
D(t)
\Big)=1
$.
Because $D=\bigcap_{t\in \mathbb{Q}} D(t)$ by the right-continuity of $({\bs X}_s)_{s\ge 0}$, we arrive at
$P_{\bs X}(D)=1$.
Consequently, Proposition \ref{NFermi} follows from Proposition  \ref{SingleEFK}. 
\end{proof}

Now, let $V$ be an interaction between electrons, which is a multiplication operator on $\ell_{\mathrm{as}}^2(\Omega^N)$ by a symmetric function $V$.  Our Hamiltonian  is given by 
$
L_V(\alpha)=L_{\mathrm{as}}(\alpha)+V.
$
Furthermore, by taking an interaction between spins and a uniform field ${\bs h}=(0, 0, 2b)$ into consideration,
we arrive at the following Hamiltonians:
\begin{align}
L_{V, b}(\alpha)=L_V(\alpha)-b\sum_{j=1}^N \sigma_j,
\end{align}
where the linear operator $\sigma_j$ is defined by 
\begin{align}
(\sigma_j f)(X^{(1)},\dots, X^{(N)})=\sigma(X^{(j)}) f(X^{(1)},\dots, X^{(N)}),\ \ f\in \ell^2_{\mathrm{as}}(\Omega^N),
\end{align}
where $\sigma(X^{(j)})$ is the spin component of $X^{(j)}$.
In general, $\alpha$ and $b$ are independent.

By using (\ref{LFK}), we obtain the following.
\begin{Prop}\label{ManyElSchFKI2}
Let $\alpha_1, \dots, \alpha_n$ be magnetic potentials.
 For every  $d\in \BbbN$, 
$F_0, F_1, \dots, F_{n-1}\in \ell^{\infty}_{\mathrm{s}}(\Omega^N),\ F\in \ell^2_{\mathrm{as}}(\Omega^N)$,  $0<t_1<t_2<\cdots <t_n$ and ${\bs X} \in \Omega_{\neq}^N$,  we have
\begin{align}
&\Big(F_0 e^{-t_1L_{V, b}(\alpha_1)} F_1 e^{-(t_2-t_1)L_{V, b}(\alpha_2)}F_2 \cdots F_{n-1}e^{-(t_n-t_{n-1}) L_{V, b}(\alpha_n)} F
\Big)({\bs X})\no
=&\Ex_{\bs X}\Bigg[ 1_{D}
F_0({\bs X}_0) F_1({\bs X}_{t_1}) \cdots F_{n-1}({\bs X}_{t_{n-1}}) F({\bs X}_{t_n})\times \no
&\ \ \ \ \times \exp\Bigg\{
\sum_{j=1}^N\sum_{\ell=1}^n \mathcal{S}_{[t_{\ell-1}, t_{\ell}]}(v, \alpha_{\ell}\, |  X^{(j)}_{\bullet})-\int_0^{t_n} V({\bs X_s})ds+b\sum_{j=1}^N \sigma(X_0^{(j)})
\Bigg\}
\Bigg].
\end{align}
\end{Prop}

\subsection{A Feynman-Kac-It\^o formula for the magnetic  Hubbard models}
Let $\alpha=(\alpha_{xy})$ be a magnetic potential.
The {\it magnetic Hubbard Hamiltonian}  is given by 
\begin{align}
H(\alpha)=\sum_{x, y\in \Lambda}\sum_{\sigma=\pm 1} (-t_{xy})e^{i\alpha_{xy}} c_{x\sigma}^* c_{y\sigma} +U\sum_{x\in \Lambda}n_{x, +1 n_{x, -1}}+\sum_{x\neq y} U_{xy} n_xn_y
-2b S_{\mathrm{tot}}^{(3)}. \label{MagHubb}
\end{align}
In order to construct a path integral formula for $H(\alpha)$, we need some preliminaries.
Let $\mathfrak{E}$ be the fermionic Fock space over $\ell^2(\Omega)$:
\begin{align}
\mathfrak{E}=\bigoplus_{N=0}^{2|\Lambda|} \bigwedge^N\ell^2(\Omega).
\end{align}
The $N$-electron space $\mathfrak{H}_{\rh}^{(N)} (\cong \ell^2_{\mathrm{as}}(\Omega^N))$
can be regarded as a subspace of $\mathfrak{E}$ in the following manner: 
\begin{align}
f_1\wedge \cdots \wedge f_N=c(f_1)^*\cdots c(f_N)^* \Omega_{\mathrm{el}}, \ \ f_1, \dots, f_N\in \ell^2(\Omega),  \label{IdnEL}
\end{align}
where $c(f)^*,\ f\in \ell^2(\Omega)$ is defined by $
c(f)^*=\sum_{X\in \Omega} f(X)c_{X}^*
$, and $\Omega_{\mathrm{el}}$ is the Fock vacuum in $\mathfrak{E}$.
Let $A=(a_{XY})_{X, Y}$ be a self-adjoint operator on $\ell^2(\Omega)$.  Under the identification (\ref{IdnEL}), we have
\begin{align}
\sum_{X, Y\in \Omega} a_{XY} c_X^*c_Y  \restriction  \mathfrak{H}_{\rh}^{(N)} 
=\dG_{\mathrm{as}, N}(A), 
\end{align}
where 
\begin{align}
\dG_{\mathrm{as}, N}(A)=
\Big( \underbrace{A\otimes 1\otimes \cdots \otimes 1}_N+
1\otimes A\otimes \cdots \otimes 1+
\cdots+ 1\otimes \cdots \otimes 1\otimes A\Big)\restriction \mathfrak{H}_{\rh}^{(N)}.\label{IdnSecQ}
\end{align}
Let $T(\alpha)=(-t_{xy}e^{i\alpha_{xy}})_{x, y}$. 
By (\ref{IdnSecQ}), we have 
\begin{align}
\sum_{x, y\in \Lambda}\sum_{\sigma=\pm 1} (-t_{xy}  e^{i\alpha_{xy}} )c_{x\sigma}^*c_{y\sigma}\restriction \mathfrak{H}_{\rh}^{(N)} 
=\dG_{\mathrm{as}, N}(T(\alpha))
=\dG_{\mathrm{as}, N}\big(h_{\mu}(\alpha)\big), \label{EqdGLa}
\end{align}
where we have used the fact $h_{0}(\alpha)=-\mu+T(\alpha)$
with $\mu(x)=-\sum_{y\in \Lambda} t_{xy}$. 
Furthermore, we obtain 
\begin{align}
n_{x, \sigma}=\dG_{\mathrm{as}, N}\big(\delta_{(x, \sigma)} \big),\ \ 
n_x=\sum_{\sigma =\pm 1} \dG_{\mathrm{as}, N}\big(\delta_{(x, \sigma)} \big),
\end{align}
where $\delta_{(x, \sigma)}$ is the multiplication operator by the function $\delta_{(x, \sigma)}(Y)
:=\delta_{\sigma\tau}\delta_{xy},\ Y=(y, \tau)\in \Omega$.
Thus, the Coulomb interaction term in (\ref{MagHubb})
can be identified with the multiplication operator $\tilde{V}$ defined by 
\begin{align}
\tilde{V}&=\tilde{V}_{\mathrm{o}}+\tilde{V}_{\mathrm{d}}
\end{align}
with 
\begin{align}
\tilde{V}_{\mathrm{d}}=U\sum_{x\in \Lambda}\sum_{i, j=1}^N \delta_{(x, +1)}^{(i)} \delta_{(x, -1)}^{(j)},\ \
\tilde{V}_{\mathrm{o}}=\sum_{\sigma, \tau=\pm 1}\sum_{x\neq y}\sum_{i, j=1}^N U_{xy} \delta_{(x, \sigma)}^{(i)}\delta_{(y, \tau)}^{(j)}, \label{VasMulti}
\end{align}
where 
\begin{align}
\delta_{(x, \sigma)}^{(j)}=\underbrace{1\otimes \cdots \otimes 1 \otimes \overbrace{\delta_{(x, \sigma)} }^{j^{\mathrm{th}}}\otimes 1\otimes  \cdots \otimes 1}_N,\ \ j=1, \dots, N.
\end{align}
Consequently, we arrive at 
\begin{align}
H(\alpha)=L_{\tilde{V}, b}(\alpha)\label{Equiv}
\end{align}
with $v=\mu$.
By Proposition \ref{ManyElSchFKI2} and (\ref{Equiv}), we obtain a Feynman-Kac-It\^o formula for $H(\alpha)$:

\begin{Thm}\label{FKIMHubbard}
Let $V=\tilde{V}+\dG_{\mathrm{as}, N}(\mu)$.
Let $\alpha_1, \dots, \alpha_n$ be  magnetic potentials.
For every  $d\in \BbbN$, 
$F_0, F_1, \dots, F_{n-1}\in \ell^{\infty}_{\mathrm{s}}(\Omega^N),\ F\in \ell^2_{\mathrm{as}}(\Omega^N)$, $0<t_1<t_2<\cdots <t_n$ and ${\bs X} \in \Omega^N_{\neq}$,  we have
\begin{align}
&\Big(F_0 e^{-t_1H(\alpha_1)} F_1 e^{-(t_2-t_1)H(\alpha_2)}F_2 \cdots F_{n-1}e^{-(t_n-t_{n-1}) H(\alpha_n)} F
\Big)({\bs X})\no
=&\Ex_{\bs X}\Bigg[ 1_{D}
F_0({\bs X}_0) F_1({\bs X}_{t_1}) \cdots F_{n-1}({\bs X}_{t_{n-1}}) F({\bs X}_{t_n})\times \no
&\ \ \ \ \times \exp\Bigg\{\sum_{j=1}^N\sum_{\ell=1}^n \mathcal{S}_{[t_{\ell-1}, t_{\ell}]}(0, \alpha_{\ell}\, |  X^{(j)}_{\bullet})-\int_0^{t_n} V({\bs X_s})ds+b\sum_{j=1}^N \sigma(X_0^{(j)})\Bigg\}
\Bigg]. \label{HubbardFKF}
\end{align}
\end{Thm}

\subsection{A  Feynman-Kac-It\^o formula for $H(\alpha)$ in the large $U$ limit}

Using a manner of proof similar to   that applied to \cite[Theorem 2.5]{Miyao}, we can prove the following proposition.
\begin{Prop}\label{MagHHULim}
Let $\alpha$ be a magnetic potential.
Assume $N=|\Lambda|-1$. We define an effective Hamiltonian $
H_{\infty}(\alpha)
$ by $
H_{\infty}(\alpha)=P_{\mathrm{G}} H^{ U=0}(\alpha)P_{\mathrm{G}}
$, where $
H^{ U=0}(\alpha)
$ is the magnetic Hubbard Hamiltonian $H(\alpha)$ with $U=0$. Then, for every $d\in \BbbN$,  we have
\begin{align}
\lim_{U\to \infty} (H(\alpha)-z)^{-1}=(H_{\infty}(\alpha)-z)^{-1}P_{\mathrm{G}},\ \ z\in \BbbC\backslash \BbbR
\end{align}
in the operator norm topology.
\end{Prop}

Let 
\begin{align}
\Omega_{\neq,  \infty}^N=\Big\{
{\bs X}\in \Omega^N\, \Big|\, x^{(i)} \neq x^{(j)}\ \ \mbox{for all $i, j\in \{1, \dots, N\}$ with $i\neq j$}
\Big\}.
\end{align}
Here, we used the following notations: ${\bs X}=(X^{(1)}, \dots, X^{(N)})
$ with $X^{(j)}=(x^{(j)}, \sigma^{(j)})$.
We set 
$D_{\infty}(\beta)=D_{\mathrm{O}, \infty}(\beta) \cap D_{\mathrm{S}}\cap D$ with
\begin{align}
D_{\mathrm{O},  \infty}(\beta)=\Big\{
{\bs m}\in (M)^N\, \Big|\, {\bs X}_s({\bs m})\in \Omega_{\neq, \infty}^N\, \mbox{ for all $s\in [0, \beta]$}
\Big\}.
\end{align}
Note that, for each ${\bs m} \in D_{\infty}(\beta)$, there are no electron encounters in the corresponding path 
$({\bs X}_t({\bs m}))_{t\in [0, \beta]}$.

\begin{Thm}\label{FKIHubbrdUInf}
Suppose that $N=|\Lambda|-1$. 
Let $\alpha_1, \dots, \alpha_n$ be  magnetic potentials.
For every  $d\in \BbbN$, 
$F_0, F_1, \dots, F_{n-1}\in \ell^{\infty}_{\mathrm{s}}(\Omega^N),\ F\in P_{\mathrm{G}}\ell^2_{\mathrm{as}}(\Omega^N)$, $0<t_1<t_2<\cdots <t_n=\beta$ and ${\bs X} \in \Omega_{\neq}^N$,  we have
\begin{align}
&\Big(F_0 e^{-t_1H_{\infty}(\alpha_1)} F_1 e^{-(t_2-t_1)H_{\infty}(\alpha_2)}F_2 \cdots F_{n-1}e^{-(t_n-t_{n-1}) H_{\infty}(\alpha_n)} F
\Big)({\bs X})\no
=&\Ex_{\bs X}\Bigg[ 1_{D_{\infty}(\beta)}
F_0({\bs X}_0) F_1({\bs X}_{t_1}) \cdots F_{n-1}({\bs X}_{t_{n-1}}) F({\bs X}_{t_n})\times \no
&\ \ \ \ \times \exp\Bigg\{\sum_{j=1}^N\sum_{\ell=1}^n \mathcal{S}_{[t_{\ell-1}, t_{\ell}]}(0, \alpha_{\ell}\, |  X^{(j)}_{\bullet}) -\int_0^{t_n} V_{{\rm o}}({\bs X_s})ds+b\sum_{j=1}^N \sigma(X_0^{(j)})
\Bigg\}
\Bigg], \label{InftyFKF}
\end{align}
where $V_{{\rm o}}=\tilde{V}_{\rm o}+\dG_{\mathrm{as}, N}(\mu)$. Here,  $\tilde{V}_{\rm o}$ is given by (\ref{VasMulti}).
\end{Thm}
\begin{proof}
By Proposition \ref{MagHHULim}, we have 
\begin{align}
&\lim_{U\to \infty} \Big(F_0 e^{-t_1H (\alpha_1)} F_1 e^{-(t_2-t_1)H(\alpha_2)}F_2 \cdots F_{n-1}e^{-(t_n-t_{n-1}) H(\alpha_n)} F
\Big)\no
&= F_0 e^{-t_1H_{\infty}(\alpha_1)} F_1 e^{-(t_2-t_1)H_{\infty}(\alpha_2)}F_2 \cdots F_{n-1}e^{-(t_n-t_{n-1}) H_{\infty}(\alpha_n)} F. \label{UInfty}
\end{align}
We denote by $1_D G_U({\bs X}_{\bullet})$ the integrand in the right hand side of (\ref{HubbardFKF}).
Then we have 
\begin{align}
\Ex_{{\bs X}}[1_DG_U({\bs X}_{\bullet})]=\Ex_{{\bs X}}[1_{D_{\infty}(\beta) }G_{U=0}({\bs X}_{\bullet})]
+\Ex_{{\bs X}}[1_{D\setminus D_{\infty}(\beta)} G_{U=0}({\bs X}_{\bullet})].
\end{align}
Because $
\lim_{U\to \infty} G_U({\bs X}_{\bullet}({\bs m}))=0
$ for all ${\bs m} \in D\setminus D_{\infty}(\beta)$, we have, by the dominated convergence theorem, 
\begin{align}
\lim_{U\to \infty}\Ex_{{\bs X}}[1_DG_U({\bs X}_{\bullet})]=\Ex_{{\bs X}}[1_{D_{\infty}(\beta) }G_{U=0}({\bs X}_{\bullet})]. \label{InftyUMain}
\end{align}
Combining (\ref{HubbardFKF}), (\ref{UInfty}) and (\ref{InftyUMain}), we obtain the desired assertion. 
\end{proof} 

\subsection{A trace formula for $H(\alpha)$} \label{TraceSect1}

First, let us construct a  complete orthonormal system (CONS) for $\ell^2_{\mathrm{as}}(\Omega^N)$.
For each ${\bs X} \in \Omega^N$,  we set $\delta_{\bs X}=\otimes_{j=1}^N\delta_{X^{(j)}} \in \ell^2(\Omega^N)$ and $e_{{\bs X}} =A_N \delta_{{\bs X}}$, where $A_N$ is the antisymmetrizer on $\ell^2(\Omega^N)$. Trivially, $\{\delta_{{\bs X}}\, |\, {\bs X} \in \Omega^N\}$ is a CONS for  $\ell^2(\Omega^N)$.
To get a CONS for $\ell_{\mathrm{as}}^2(\Omega^N)$, we need some preliminaries.
For all $\tau\in \mathfrak{S}_n$, we know that $e_{\tau {\bs X}}=\mathrm{sign}(\tau) e_{{\bs X}}$. Taking this fact into consideration, we introduce an  equivalence relation in $\Omega_{\neq}^N$ as follows:
Let ${\bs X}, {\bs Y} \in \Omega_{\neq}^N$. If there exists a $\tau\in \mathfrak{S}_N$ such that ${\bs Y}=\tau {\bs X}$, then we write ${\bs X} \equiv {\bs Y}$. Let $[{\bs X}]$ be the equivalence class to which ${\bs X}$
belongs. We will often abbreviate $[{\bs X}]$ to ${\bs X}$ if no confusion occurs.
The quotient set $\Omega_{\neq}^N/\equiv $ is denoted by $[\Omega_{\neq}^N]$. 
Then we readily confirm that $\{e_{\bs X}\, |\, {\bs X} \in [{\Omega_{\neq}^N}]\}$ is a CONS for $\ell^2_{\mathrm{as}}(\Omega^N)$.

\begin{define}{\rm 
Let $P_{\neq}$ be the orthogonal projection from $\ell^2(\Omega^N)$ to $\ell^2(\Omega_{\neq}^N)$. 
We define a self-adjoint operator  on $\ell^2(\Omega_{\neq}^N)$
 by 
\begin{align}
{\bs L}&=P_{\neq}L(\alpha=0) |_{v=0} P_{\neq}.
\end{align}
Here, we note that by recalling (\ref{NELHAMI}),  $L(\alpha=0) |_{v=0} $ can be explicitly expressed as 
$
h_0(0)\otimes 1\otimes \cdots\otimes 1+\cdots +1\otimes \cdots \otimes 1\otimes h_0(0).
$
Fix ${\bs X}\in \Omega^N_{\neq}$, arbitrarily.
A permutation $\tau\in \mathfrak{S}_N$ is called {\it dynamically allowed}
associated with ${\bs X}$ if there is an $n\in\BbbN_0$
such that 
\begin{align}
\la \delta_{{\bs X}}|{\bs L}^n \delta_{\tau{\bs X}}\ra \neq 0. \label{DyAlll1}
\end{align}
The set of all dynamically allowed permutations associated with ${\bs X}$ 
is denoted by $\mathfrak{S}_N({\bs X})$.
As we will see below,  the dynamically allowed permutations play important roles. 
}
\end{define}

We give  a useful characterization of (\ref{DyAlll1}).
For each ${\bs X}=(X^{(j)})_{j=1}^N, {\bs Y}=(Y^{(j)})_{j=1}^N\in \Omega^N_{\neq}$, we define a distance
between ${\bs X}$ and ${\bs Y}$ by
$\|{\bs X}-{\bs Y}\|_{\infty}=\max_{j=1, \dots, N}\|x^{(j)}-y^{(j)}\|_{\infty}$, where $x^{(j)}$ (resp. $y^{(j)}$)
 is the spatial component of $X^{(j)}$ (resp. $Y^{(j)}$).  We say that ${\bs X}$ and ${\bs Y}$ are {\it  neighbors}
  if $\|{\bs X}-{\bs Y}\|_{\infty}=1$. A pair  $\{ {\bs X}, {\bs Y}\} \in \Omega^N_{\neq} \times \Omega^N_{\neq}$ is called an {\it edge},  if ${\bs X}$ and ${\bs Y}$ are neighbors. 
  We say that a sequence $({\bs X}_i)_{i=1}^m \subset \Omega_{\neq}^N$ is a {\it path}, if 
  $\{{\bs X}_i, {\bs X_{i+1}}\}$ is an edge for all $i$.
For each edge $\{{\bs X}, {\bs Y}\}$, we define a linear operator acting in $\ell^2(\Omega_{\neq}^N)$
by 
\begin{align}
Q({\bs X}, {\bs Y})= |\delta_{{\bs X}}\ra \la \delta_{{\bs Y}}|.
\end{align}
\begin{lemm}\label{EquivPath}
The following {\rm (i)} and {\rm (ii)} are mutually equivalent:
 \begin{itemize}
 \item[{\rm (i)}] $\tau$ is dynamically allowed associated with  ${\bs X}$;
 \item[{\rm (ii)}] there is a path $({\bs X}_i)_{i=1}^m$ such that 
 \begin{itemize}
 \item[] ${\bs X}_1={\bs X}$ and ${\bs X}_m=\tau {\bs X}$;
 \item[] $
 \la \delta_{{\bs X}}|Q({\bs X}_1, {\bs X}_2) Q({\bs X}_2, {\bs X}_3) \cdots Q({\bs X}_{m-1}, {\bs X}_m) \delta_{\tau {\bs X}}\ra >0.
 $
 \end{itemize}
 \end{itemize}
 \end{lemm}
 \begin{proof}
 By using \eqref{EqdGLa}, we can express ${\bs L}$ as 
 \begin{align}
 {\bs L}=\sum_{\{{\bs X}, {\bs Y}\}} C_{{\bs X}, {\bs Y}}Q({\bs X}, {\bs Y})
+\mathcal{D}, \label{OfDD}
 \end{align}
 where $
 C_{{\bs X}, {\bs Y}}<0
 $  for each edge $\{{\bs X}, {\bs Y}\}$ and $\mathcal{D}$ is some multiplication operator.
 By using this formula,  we can readily confirm this lemma.
 \end{proof}

We introduce  a  subspace of $\ell^2(\Omega^N)$ by
\begin{align}
\ell^2_{\mathrm{s}}(\Omega_{\neq}^N)=\big\{F\in \ell^2(\Omega_{\neq}^N)\, \big|\, \mbox{$F$ is symmetric}
\big\}.
   \end{align}   
The subspace $
\ell^2_{\mathrm{s}}(\Omega_{\neq}^N)
$ describes  wave functions for a system of $N$ {\it hard-core bosons}.

For $\beta>0$, let
\begin{align}
D_{\mathrm{P}}(\beta)=\big\{{\bs m}\in (M)^N\, \big|\, \mbox{$\exists\tau \in \mathfrak{S}_N({\bs X}_0({\bs m}))$ such that ${\bs X}_{\beta}({\bs m})=\tau {\bs X}_0({\bs m})$} \big\}. \label{DefDPP}
\end{align}
We set  
\begin{align}
L_{\beta}=D \cap D_{\mathrm{P}}(\beta).
\end{align}
Let us introduce a random variable on $L_{\beta}$ by 
\begin{align}
(-1)^{\pi({\bs X}_{\beta}({\bs m}))}=\mathrm{sgn}(\tau),
\end{align}
where $\tau$ is given in (\ref{DefDPP}).
For all ${\bs m} \in L_{\beta}$, the electron configuration at $\beta$, i.e., ${\bs X}_{\beta}({\bs m})$ is a permutation of the configuration at $t=0$;  $(-1)^{\pi({\bs X}_{\beta}({\bs m}))}$ is the parity of the permutation.

\begin{Thm}\label{TrFKIMHubbard}
Let  $\alpha_1, \dots, \alpha_n$ be    magnetic potentials.
Suppose that 
$F_0, F_1, \dots, F_{n-1}\in \ell^{\infty}_{\mathrm{s}}(\Omega^N)
$.
Then,  for all $d\in \BbbN$ and $\beta >0$, there exists  a probability measeure $\Me$ on $L_{\beta}$ such that, for $0<t_1<t_2<\cdots <t_{n-1}<\beta$,
\begin{align}
&\Tr_{\ell^2_{\mathrm{as}}(\Omega^N)}\Big[F_0 e^{-t_1H(\alpha_1)} F_1 e^{-(t_2-t_1)H(\alpha_2)}F_2 \cdots F_{n-1}e^{-(\beta-t_{n-1}) H(\alpha_n)} 
\Big] \bigg/\Tr_{\ell_{\mathrm{s}}^2(\Omega^N_{\neq})} \Big[
e^{-\beta {\bs L}}
\Big]
\no
=&
\int_{L_{\beta}}d \Me
F_0({\bs X}_0) F_1({\bs X}_{t_1}) \cdots F_{n-1}({\bs X}_{t_{n-1}}) (-1)^{\pi({\bs X}_{\beta})}\times \no
&\ \ \ \ \times \exp\Bigg\{\sum_{j=1}^N\sum_{\ell=1}^n \mathcal{S}_{[t_{\ell-1}, t_{\ell}]}(0, \alpha_{\ell}\, |  X^{(j)}_{\bullet}) -\int_0^{t_n} V({\bs X_s})ds+b\sum_{j=1}^N \sigma(X_0^{(j)})
\Bigg\}
\end{align}
with $t_n=\beta$.
\end{Thm}
\begin{proof}
We divide the proof into two parts.

{\bf Step 1.}
Let $K_n
=F_0 e^{-t_1H(\alpha_1)} F_1 e^{-(t_2-t_1)H(\alpha_2)}F_2 \cdots F_{n-1}e^{-(\beta-t_{n-1}) H(\alpha_n)} 
$. 
Fix ${\bs X}\in \Omega_{\neq}^N$, arbitrarily.
We claim  that if $\tau$ is {\it not} dynamically allowed associated with ${\bs X}$, then it holds that, for all $n\in \BbbN$ and $
0<t_1<t_2<\cdots <t_{n-1}<\beta
$,
\begin{align}
\la \delta_{{\bs X}}|K_{n} \delta_{\tau{\bs X}}\ra=0. \label{DAConc}
\end{align}
 First, let us consider the case where $U_{xy}=0$ for all $x, y$ and $b=0$. 
  Set ${\bs L}(\alpha)= P_{\neq} L (\alpha)P_{\neq}$.  For simplicity, suppose that   $n=2$.
  Because $F_0$ and $F_1$ are diagonal, we have 
  \begin{align}
  \la \delta_{{\bs X}}|F_0 (-{\bs L}(\alpha_1))^{n_1} F_1 (-{\bs L}(\alpha_2))^{n_2}\delta_{\tau{\bs X}}\ra=0 \label{DyAll2}
  \end{align}
    for all $n_1, n_2\in \BbbN_0$. Indeed, by using the formula \eqref{OfDD}, 
    the equation (\ref{DyAll2}) follows from the property 
    \begin{align}
    \la \delta_{{\bs X}}|Q({\bs X}_1, {\bs X}_2) Q({\bs X}_2, {\bs X}_3) \cdots Q({\bs X}_{m-1}, {\bs X}_m) \delta_{\tau {\bs X}}\ra =0
    \end{align}
    for any path $({\bs X}_i)_{i=1}^m$. But this is obvious from Lemma \ref{EquivPath}.
   Using (\ref{DyAll2}), we can prove (\ref{DAConc}) as follows:
   \begin{align}  
     \la \delta_{{\bs X}}|K_{n} \delta_{\tau{\bs X}}\ra 
     =\sum_{n_1=1}^{\infty}\sum_{n_2=1}^{\infty}\frac{t_1^{n_1} (t_2-t_1)^{n_2}}{n_1!n_2!}\la \delta_{{\bs X}}|F_0 (-{\bs L}(\alpha_1))^{n_1} F_1 (-{\bs L}(\alpha_2))^{n_2}\delta_{\tau{\bs X}}\ra=0.
      \end{align}
Similarly, we can prove the assertion for general $n$, provided that $U_{xy}=0$ and $b=0$.

Next, let us consider the case where $U_{xy}\neq 0$ but $b=0$. 
Again, we will consider the case $n=2$ for simplicity.
By the Trotter-Kato formula, we have
\begin{align}
&\la \delta_{{\bs X}}|K_{n} \delta_{\tau{\bs X}}\ra\no
=&\lim_{N_1\to \infty}\lim_{N_2\to \infty}
\big\la \delta_{\bs X}| F_0\big(
e^{-t_1 {\bs L}(\alpha_1)/N_1} e^{-t_1\tilde{V}/N_1}
\big)^{N_1 }  F_1\times \no
&\ \ \ \ \times\big (
e^{-(t_2-t_1) {\bs L}(\alpha_2)/N_2} e^{-(t_2-t_1)\tilde{V}/N_2}
\big)^{N_1 }\delta_{\tau {\bs X}}\big\ra. \label{TKN2}
\end{align}
 By applying  the claim for the case where  $U_{xy}\equiv 0$ and $b=0$,  we know that the right hand side of (\ref{TKN2}) equals zero.
 Similarly, we can prove the assertion for general $n$.
 
 Finally, we will consider the case where $b\neq 0$. But the proof of this case is similar to that of the case where $U_{xy}\neq 0$ and $b=0$.

{\bf Step 2.}
By Theorem \ref{FKIMHubbard} and (\ref{DAConc}), we get 
\begin{align}
\Tr_{\ell^2_{\mathrm{as}}(\Omega^N)}[K_n]
=&\sum_{{\bs X} \in [\Omega_{\neq}^N]} \la e_{{\bs X}}\, |\, K_ne_{{\bs X}}\ra\no
=& \sum_{{\bs X} \in [\Omega_{\neq}^N]} \sum_{\tau \in \mathfrak{S}_N} \frac{\mathrm{sgn}(\tau)}{N!}
\la \delta_{{\bs X}}\, |\, K_n \delta_{\tau {\bs X}}\ra\no
=& \sum_{{\bs X} \in [\Omega_{\neq}^N]} \sum_{\tau \in \mathfrak{S}_N({\bs X})} \frac{\mathrm{sgn}(\tau)}{N!}
\la \delta_{{\bs X}}\, |\, K_n \delta_{\tau {\bs X}}\ra\no
=& \sum_{{\bs X} \in [\Omega_{\neq}^N]} \sum_{\tau \in \mathfrak{S}_N({\bs X})} \frac{\mathrm{sgn}(\tau)}{N!}
\Ex_{{\bs X}} \Big[\mathcal{K}_n({\bs X}_{\bullet})
1_{\{{\bs X}_{\beta}=\tau{\bs X}\}\cap D}
\Big], \label{TraF1}
\end{align}
where  $\mathcal{K}_n$ is given by 
\begin{align}
\mathcal{K}_n({\bs X}_{\bullet})
=&
F_0({\bs X}_0) F_1({\bs X}_{t_1}) \cdots F_{n-1}({\bs X}_{t_{n-1}}) \times \no
&\times e^{\sum_{j=1}^N\sum_{\ell=1}^n \mathcal{S}_{[t_{\ell-1}, t_{\ell}]}(v, \alpha_{\ell}\, |  X^{(j)}_{\bullet})}e^{ -\int_0^{t_n} V({\bs X_s})ds+b\sum_{j=1}^N \sigma(X_0^{(j)})
}.
\end{align}
Let us define a probability measure on  $L_{\beta}$ by 
\begin{align}
\nu_{\beta}(B)=\sum_{{\bs X} \in [\Omega_{\neq}^N]} \sum_{\tau\in \mathfrak{S}_N({\bs X})} \frac{1}{N!} 
P_{\bs X} \Big(
B\cap \{{\bs X}_{\beta}=\tau{\bs X}\} \cap D
\Big) \Big/ {\mathrm{Norm.}}, \label{Measurenu}
\end{align}
where $\mathrm{Norm.}$ is the normalization constant, which can be computed as follows:
\begin{align}
\mathrm{Norm.}&= \sum_{{\bs X} \in [\Omega_{\neq}^N]} \sum_{\tau\in \mathfrak{S}_N({\bs X})}\frac{1}{N!}
P_{\bs X} (\{{\bs X}_{\beta}=\tau {\bs X}\}\cap D)\no
&=\sum_{{\bs X}\in [\Omega_{\neq}^N]}
\sum_{\tau\in \mathfrak{S}_N({\bs X})}\frac{1}{N!}
\la \delta_{\bs X}\, |e^{-\beta {\bs L}} \delta_{\tau {\bs X}}\ra\no
&=\Tr_{\ell^2_{\mathrm{s}}(\Omega^N_{\neq })} \Big[e^{-\beta {\bs L}}\Big]. \label{Norm.}
\end{align}
By combining (\ref{TraF1}) and  (\ref{Measurenu}), we obtain the desired assertion with $\mathrm{sgn}(\tau)=(-1)^{\pi({\bs X}_{\beta})}$. 
\end{proof}

\subsection{A trace formula for $H_{\infty}(\alpha)$}

First, we note that a natural CONS for $P_{\mathrm{G}} \ell_{\mathrm{as}}^2(\Omega^N)$ is 
$\{e_{\bs X}\, |\, {\bs X} \in [\Omega_{\neq, \infty}^{N}]\}$, where $[\Omega_{\neq, \infty}^N]$
 is the quotient set $\Omega_{\neq, \infty}^N/\equiv$, see Section \ref{TraceSect1}.
 Let ${\bs L}_{\infty}=P_{\mathrm{G}} {\bs L} P_{\mathrm{G}}$. Even if $U=\infty$, we can  define the dynamically allowed permutations: Fix ${\bs X}\in \Omega^N_{\neq}$, arbitrarily.
A permutation $\tau\in \mathfrak{S}_N$ is called {\it dynamically allowed}
associated with ${\bs X}$ if there is an $n\in\BbbN_0$
such that 
\begin{align}
\la \delta_{{\bs X}}|{\bs L}_{\infty}^n \delta_{\tau{\bs X}}\ra  \neq 0. \label{DyAll1}
\end{align}
The set of all dynamically allowed permutations associated with ${\bs X}$ 
is denoted by $\mathfrak{S}_{N, \infty}({\bs X})$.
Because there are no encounter of electrons, the trajectories are very simplified.\footnote{
Reader can readily confirm this fact for one-dimensional chain;
in fact, the dynamically allowed permutation can only be the identity in this case.
}
In particular, in the case where $N=|\Lambda|-1$, 
we can check that if $\tau$ is dynamically allowed, then $\tau$ is always even: $\mathrm{sgn}(\tau)=1$ \cite{AL}.

For $\beta>0$, let
\begin{align}
D_{\mathrm{P}, \infty}(\beta)=\big\{{\bs m}\in (M)^N\, \big|\, \mbox{$\exists\tau \in \mathfrak{S}_{N, \infty}({\bs X}_0({\bs m}))$ such that ${\bs X}_{\beta}({\bs m})=\tau {\bs X}_0({\bs m})$} \big\}. \label{DefDP}
\end{align}
We define an event by 
\begin{align}
L_{\beta, \infty}=D_{\infty}(\beta) \cap D_{\mathrm{P}, \infty}(\beta).
\end{align}
As before, we can define the random variable $(-1)^{\pi({\bs X}_{\beta}({\bs m}))}$ for all ${\bs m} \in L_{\beta, \infty}$. But because each dynamically allowed permutation is  even   in the case where $N=|\l\Lambda|-1$ and $U=\infty$, it holds that $(-1)^{\pi({\bs X}_{\beta}({\bs m}))}=1$. 
Taking this fact into account and 
using arguments similar to those in the proof of Theorem \ref{TrFKIMHubbard}, we can prove the following.

\begin{Thm}\label{FKIUINFINITY}
Suppose that $N=|\Lambda|-1$.
Let  $\alpha_1, \dots, \alpha_n$ be    magnetic potentials.
Suppose that 
$F_0, F_1, \dots, F_{n-1}\in \ell^{\infty}_{\mathrm{s}}(\Omega^N)
$.
Then, for every $d\in \BbbN$ and $\beta >0$, there exists  a probability measeure $\Mei$ on $L_{\beta, \infty}$ such that, for $0<t_1<t_2<\cdots <t_{n-1}<\beta$,
\begin{align}
&\Tr_{P_{\mathrm{G}}\ell^2_{\mathrm{as}}(\Omega^N)}\Big[F_0 e^{-t_1H_{\infty}(\alpha_1)} F_1 e^{-(t_2-t_1)H_{\infty}(\alpha_2)}  F_2\cdots F_{n-1}e^{-(\beta-t_{n-1}) H_{\infty}(\alpha_n)} 
\Big]
\no
=& 
\Tr_{P_{\mathrm{G}}\ell_{\mathrm{s}}^2(\Omega^N_{\neq})} \Big[
e^{-\beta {\bs L}_{\infty}}
\Big]
\int_{L_{\beta, \infty}}d \Mei
F_0({\bs X}_0) F_1({\bs X}_{t_1}) \cdots F_{n-1}({\bs X}_{t_{n-1}})\times \no
&\ \ \ \ \times \exp\Bigg\{\sum_{j=1}^N\sum_{\ell=1}^n \mathcal{S}_{[t_{\ell-1}, t_{\ell}]}(0, \alpha_{\ell}\, |  X^{(j)}_{\bullet}) -\int_0^{t_n} V_{{\rm o}}({\bs X_s})ds+b\sum_{j=1}^N \sigma(X_0^{(j)}) \Bigg\}
\label{FKITrUInf}
\end{align}
with $t_n=\beta$, where $V_{{\rm o}}$ is defined in Theorem \ref{FKIHubbrdUInf}.
\end{Thm}

\section{Trace formulas for free Bose field}\label{TraceBoseGene}
\setcounter{equation}{0}

\subsection{Preliminaries}\label{BosePre}

Let $\mathfrak{X}_r$ be a real separable  Hilbert space equipped with the inner product $\la \cdot |\cdot \ra_{\mathfrak{X}_r}$.
Let $\{\phi(f)\, |\, f\in \mathfrak{X}_r\}$ be the Gaussian random process indexed by $\mathfrak{X}_r$,  and let 
$(Q, \mathcal{F}, \mu)$ be its underlying probability space. Note that 
\begin{align}
\int_Qd\mu \phi(f)\phi(g)=\frac{1}{2} \la f|g\ra_{\mathfrak{X}_r},\ \ f, g\in \mathfrak{X}_r.
\end{align}

  Let $A$ be a positive self-adjoint operator acting in $\mathfrak{X}_r$. Suppose that there exists a constant $a_0>0$ such that $A\ge a_0$ (i.e., $\la f|Af\ra\ge a_0\|f\|^2$ for all $f\in \D(A^{1/2})$ ).
  For each $s\in \BbbR$, we introduce an inner product $\la \cdot |\cdot \ra_s$ on $\D(A^{s/2})$ by 
  $\la f|g\ra_s=\la A^{s/2}f|A^{s/2}g\ra,\ f, g\in \D(A^{s/2})$.
  For $s>0$, $(\D(A^{s/2}), \la \cdot|\cdot\ra_s)$ becomes a Hilbert space, which is denoted by $\mathfrak{X}_{r, s}$. For $s<0$,
    we denote by $\mathfrak{X}_{r, s}$ the completion of $\mathfrak{X}_{r}(=\D(A^{s/2}))$ in the norm $\|\cdot\|_s:=\la \cdot |\cdot \ra_s^{1/2}$.  The dual space of $\mathfrak{X}_{r, s}$ can be identified with $\mathfrak{X}_{r, -s}$ through the bilinear form such that  ${}_{-s}\la f|g\ra_s=\la f|g\ra_{\mathfrak{X}_r},\ f\in \mathfrak{X}_{r, -s} \cap \mathfrak{X}_r,\ g\in \mathfrak{X}_{r, s}\cap \mathfrak{X}_r$.
    
    In what follows, we  assume the following:
    \begin{itemize}
    \item[{\bf (A)}] For some $\gamma_0>0$, $A^{-\gamma_0}$ is in the trace class.
    \end{itemize}
Choose $\gamma>\gamma_0$, arbitrarily. Clearly, the embedding mapping of $\mathfrak{X}_r$ into $\mathfrak{X}_{r, -\gamma}$ is in the Hilbert-Schmidt class. Thus, by applying  \cite[Proposition 5.1]{Gross},
we can take $Q=\mathfrak{X}_{r, -\gamma}$ and $\phi(f)= {}_{-\gamma}\la \phi|f\ra_{\gamma},\ 
\phi\in Q,\ f\in \mathfrak{X}_{r, \gamma}$.

We denote by $\mathfrak{X}$ the complexification of $\mathfrak{X}_r$. Then each element $f$ in $\mathfrak{X}$
 can be expressed as $f=f_1+if_2,\ \ f_1, f_2\in \mathfrak{X}_r$. Now we define $\phi(f),\ f\in \mathfrak{X}$
  by $\phi(f)=\phi(f_1)+i\phi(f_2)$. Trivially, we have 
  $\int_Qd\mu \phi(\overline{f}) \phi(g)=\frac{1}{2} \la f|g\ra_{\mathfrak{X}},\ f,g\in \mathfrak{X}$, where 
  $\overline{f}=f_1-if_2$.

  For each $f\in \mathfrak{X}$, we define a symmetric operator $\Phi_{\mathrm{S}}(f)$ by 
  \begin{align}
  \Phi_{\mathrm{S}}(f)=\frac{1}{\sqrt{2}} (a(f)^*+a(f)).
  \end{align}
  It is well-known that $\Phi_{\mathrm{S}}(f)$ is essentially self-adjoint on $\Ffin(\mathfrak{X})$. We denote its closure by the same symbol.  $\Phi_{\mathrm{S}}(f)$ is called {\it Segal's field operator}.

  There is a useful  identification between  $L^2(Q, d\mu)$ and  $\Fock(\mathfrak{X})$; namely, 
  there exists a unitary operator $U$ from $\Fock(\mathfrak{X})$ onto $L^2(Q, d\mu)$ satisfying the following (i)-(iii) \cite{Simon}:
  \begin{itemize}
  \item[(i)] $U\Omega_{\mathrm{b}}=1$, where $\Omega_{\mathrm{b}}$ is the Fock vacuum in $\Fock(\mathfrak{X})$;
  \item[(ii)] $U a(f_1)^*\cdots a(f_n)^* \Omega_{\mathrm{b}}=2^{n/2}: \phi(f_1)\cdots \phi(f_n):,\ \ f_1, \dots, f_n\in \mathfrak{X}$,  where 
  $
  : \phi(f_1)\cdots \phi(f_n):
  $ indicates the Wick product;
  \item[(iii)] $U\Phi_{\mathrm{S}}(f) U^{-1}=\phi(f),\ f\in \mathfrak{X}$.
  \end{itemize}
 Let $B$ be a positive self-adjoint operator on $\mathfrak{X}_r$. We define a linear operator $\dG_{\mathrm{s}}(B)$ acting in $\Fock(\mathfrak{X})$ by 
  \begin{align}
  \dG_{\mathrm{s}}(B) \restriction \otimes_{\mathrm{s}}^n \mathfrak{X}
  =\underbrace{B\otimes 1\otimes \cdots \otimes 1}_n+1\otimes B \otimes 1 \otimes  \cdots\otimes 1 +\cdots +1\otimes \cdots\otimes 1 \otimes B. \label{Bose2ndQ}
  \end{align}
  $\dG_{\mathrm{s}}(B)$ is called the {\it second quantization} of $B$. $\dG_{\mathrm{s}}(B)$ is  positive and self-adjoint \cite{ BR, Simon}.
Let $\{e_n\}_{n=1}^{\infty}\subset \mathfrak{X}_r$ be a CONS of $\mathfrak{X}$.
Suppose that $B$ is diagonal with respect to $\{e_n\}_{n=1}^{\infty}$:
$Be_n=\lambda_n e_n,\ n\in \BbbN$. Then we have 
\begin{align}
\sum_{n=1}^{\infty}\lambda_n a(e_n)^*a(e_n)=U^{-1} \dG_{\mathrm{s}}(B) U \label{Tubutubu}
\end{align}
 on the  dense subspace $
 \{\Psi=(\Psi_n)_{n=0}^{\infty}\in \Ffin(\mathfrak{X})\, |\, \Psi_n\in \otimes_{\mathrm{alg}}^n\D(B)\}
 $, where $\otimes_{\mathrm{alg}}$ indicates the incompleted tensor product.

\subsection{A trace formula for free Euclidean field I}

Let $A$ be the linear operator given in the previous subsection.
Note that $A$ can be naturally extended to the operator on  $\mathfrak{X}$. We denote the extension by the same symbol.
By the assumption {\bf (A)}, we know that $e^{-\beta A}$ is in the trace class as an operator on $\mathfrak{X}$  for all $\beta >0$. Accordingly,  $e^{-\beta \dG_{\mathrm{s}}(A)}$ is in the trace class as an operator on $\Fock(\mathfrak{X})$ for all $\beta >0$ satisfying 
\begin{align}
\mathcal{Z}_{A}(\beta):=\Tr_{\Fock(\mathfrak{X})}\Big[e^{-\beta \dG_{\mathrm{s}}(A)}\Big]=\frac{1}{\det(1-e^{-\beta A})},
\end{align}
where $\det(\cdots)$ is the determinant \cite{ReSi4}.

We set $Q_{\beta}=C_{\mathrm{P}}([0, \beta]; Q)$, the space of continuous loop of $Q$ with parameter  space $[0, \beta]$. For each $\Phi\in Q_{\beta}$, the value of $\Phi$ at $t$ is denoted by $\Phi_t\in Q$.
Let $\mathcal{F}_{\beta}$ be the Borel field on $Q_{\beta}$ generated by $\{\Phi_t(f)|\, f\in \mathfrak{X}_r, t\in [0, \beta]\}$.

\begin{Prop}[\cite{Arai, HK}]
There exists a probability measure $\mu_{\beta}$ on $(Q_{\beta}, \mathcal{F}_{\beta})$ such that 
$\{\Phi_t(f)\, |\, f\in \mathfrak{X}_{r, \gamma},\, t\in [0, \beta]\}$ is a family of jointly Gaussian random processes on $(Q_{\beta}, \mathcal{F}_{\beta}, \mu_{\beta})$ with covariance 
\begin{align}
\int_{Q_{\beta}} d\mu_{\beta} \Phi_s(f)\Phi_t(g)=\frac{1}{2} \big\la f \big|\big(1-e^{-\beta A}\big)^{-1}
\big (
e^{-(\beta-|t-s|)A}+e^{-|t-s|A}
\big)g\big\ra_{\mathfrak{X}_r},\ \ f, g\in \mathfrak{X}_r.
\end{align}
Let $G_0, \dots, G_n$ be bounded measurable functions on $\BbbR^m$. 
For  $0<t_1<t_2<\cdots<t_n<\beta$, we have 
\begin{align}
&\Tr_{\Fock(\mathfrak{X})} \Big[
G_0^{\rf} e^{-t_1\dG_{{\rm s}}(A)} G_1^{\rf} e^{-(t_2-t_1)\dG_{{\rm s}}(A)}G_2^{\rf}\cdots G_n^{\rf}e^{-(\beta-t_n)\dG_{{\rm s}}(A)}
\Big]\Bigg/\mathcal{Z}_{A}(\beta) \no
=& \int_{Q_{\beta}} d\mu_{\beta} G_0^{0} G_1^{t_1} \cdots G_n^{t_n},
\end{align}
where
\begin{align}
G_j^{\rf}=G_j\Big(\Phi_{\mathrm{S}}(f_1^{(j)}), \dots,  \Phi_{\mathrm{S}}(f_m^{(j)})\Big),\ \ \
G^t_j=G_j\Big(\Phi_{t}(f_1^{(j)}), \dots,  \Phi_{t}(f_m^{(j)})\Big)
\end{align}
for $f_1^{(j)}, \dots, f_m^{(j)}\in \mathfrak{X}_{r, \gamma}$.
\end{Prop}

\subsection{A trace formula for free Euclidean field II}
Let $L_r^2(0, \beta)$ be the real Hilbert space of real-valued measurable functions in $L^2(0, \beta)$
 and set $\mathfrak{X}_r^{\beta}=L_r^2(0, \beta)\otimes \mathfrak{X}_r$.
 Let $\Delta_{\mathrm{P}}$ be the periodic Laplacian acting in $L^2(0, \beta)$. 
 We introduce a norm of $\mathfrak{X}_r^{\beta}$ by 
 \begin{align}
 \|f\|_{-1, \beta}^2=\frac{1}{2}\Bigg\|
 \sqrt{\frac{1 \otimes A^2}
 {(-\Delta_{\mathrm{P}})\otimes 1 +1 \otimes A^2}
 } f
 \Bigg\|^2_{\mathfrak{X}_r^{\beta}},\ \ f\in \mathfrak{X}_r^{\beta}. \label{HenNorm}
 \end{align}
We denote by $\mathfrak{X}_{-1, r}^{\beta}$ the completion of $\mathfrak{X}_r^{\beta}$ by the norm $\|\cdot\|_{-1, \beta}$. The following formula will be useful:
\begin{align}
\la \delta_s\otimes f |\delta_t\otimes  g\ra_{-1, \beta}=
\frac{1}{2} \big\la f \big|\big(1-e^{-\beta A}\big)^{-1}
\big (
e^{-(\beta-|t-s|)A}+e^{-|t-s|A}
\big)g\big\ra_{\mathfrak{X}_r},\ \ f, g\in \mathfrak{X}_{r}, \label{KeyEqCo}
\end{align}
where $\delta_t$ is the Dirac delta function.
Let $f$ be an $\mathfrak{X}_r$-valued measurable function on $[0, \beta]$ such that 
$
\int_0^{\beta} \|f(t)\|^2_{-\gamma}dt<\infty
$.
We define the smeared random variable by $\Phi(f)=\int_0^{\beta} {}_{-\gamma}\la
\Phi_t|f(t)\ra_{\gamma}dt
$. Using (\ref{KeyEqCo}), we obtain  $
\int_{Q_{\beta}} d\mu_{\beta} \Phi(f)\Phi(g)=\la f|g\ra_{-1, \beta},\ f,g\in \mathfrak{X}_{-1, r}^{\beta}
$, where the inner product $\la \cdot |\cdot \ra_{-1, \beta}$ is naturally obtained from (\ref{HenNorm}). Therefore, $\{\Phi(f)\, |\, f\in \mathfrak{X}_{-1, r}^{\beta}\}$ becomes a  Gaussian mean zero random process indexed by $\mathfrak{X}_{-1, r}^{\beta}$; its underlying probability space is 
$(Q_{\beta}, \mathcal{F}_{\beta}, \mu_{\beta})$.  

Let $\beta>0$. By using the fact that $\coth x>0$, provided that $x>0$, we define a self-adjoint operator $B(\beta)$ on $\mathfrak{X}_r$ by 
\begin{align}
B(\beta)= \sqrt{\coth \frac{\beta A}{2}}.
\end{align}
Because $A\ge a_0$, we readily see that $1\le B(\beta)\le \sqrt{\coth a_0}$. Furthermore, by using the elementary fact $ \sqrt{\coth x}-1\le De^{-2x}$ with $D=2^{-a_0}/(e^{a_0}-e^{-a_0})$ for all $x\in [a_0, \infty)$, we obtain that 
$0\le \Tr_{\mathfrak{X}_r}[B(\beta)-1] \le D\Tr_{\mathfrak{X}_r} [e^{-2\beta A}]<\infty$,
which implies that $B(\beta)-1$ is in the trace class. In particular, $B(\beta)-1 $ is in the Hilbert-Schmidt class as well. By Shale's theorem \cite[Theorem I. 23]{Simon}, there exists a probability measure $\mu_{B(\beta)}$ on $(Q, \mathcal{F})$ mutually absolutely continuous to $\mu$ such that 
\begin{align}
\int_Q d\mu_{B(\beta)} e^{i\phi(f)}=\int_Q d\mu e^{i\phi(B(\beta)f)}=e^{-\|B(\beta) f\|_{\mathfrak{X}_r}^2/4}
\end{align}
and $d\mu_{B(\beta)}=G_{\beta}d\mu$ with $G_{\beta}\in L^p(Q, d\mu)$ for some $p>1$ and $G_{\beta}^{-1}\in L^q(Q, d\mu)$ for some $q>1$.

For $t\in [0, \beta)$, we define a linear operator $j_t$ from $\mathfrak{X}_{r, \gamma}$ to $\mathfrak{X}_{-1, r}^{\beta}$ by 
\begin{align}
j_t f=\delta_t\otimes f,\ \ f\in \mathfrak{X}_{r, \gamma}.
\end{align}
By (\ref{KeyEqCo}), we have 
\begin{align}
\| j_t f\|_{-1, \beta}^2=\frac{1}{2}\|B(\beta) f\|_{\mathfrak{X}_r}^2=:|\! |\! | f |\! |\! |_{\beta}^2.\label{Def3Norm}
\end{align}
Let $\mathfrak{X}_{r, \gamma}^{(\beta)}$ be the completion of $\mathfrak{X}_{r, \gamma}$ by the norm 
$
|\! |\! | \cdot |\! |\! |_{\beta}
$. From (\ref{Def3Norm}), it follows that $j_t$ is the  isometry from $\mathfrak{X}_{r, \gamma}^{(\beta)}$ into 
$\mathfrak{X}_{-1, r}^{\beta}$.
Now we define a linear operator  $ J_t: L^2(Q, d\mu_{B(\beta)}) \to L^2(Q_{\beta}, d\mu_{\beta})$ by 
$J_t=\Gamma(j_t)$, where for each contraction operator $C$, $\Gamma(C)$ is defined by $
\Gamma(C) : \phi(f_1)\cdots \phi(f_n):=: \Phi(Cf_1)\cdots \Phi(Cf_n):,\ f_1, \dots, f_n \in \mathfrak{X}_{r, \gamma}^{(\beta)}
$ and $\Gamma(C)1=1$.   
Note that the mapping $t\to J_t$ is  strongly continuous,  and $J_t$ is an isometry. In addition, we have the following \cite{Arai, Simon}:
\begin{itemize}
\item $J_t$ is positivity preserving;
\item $(J_t F)(\Phi)=F(\Phi_t),\ F\in L^2(Q, d\mu_{B(\beta)})$; thus, 
the mapping $t\to F(\Phi_t)$ is continuous in $L^2(Q_{\beta}, d\mu_{\beta})$;
\item $J_t$ can be extended to a contraction from $L^p(Q, d\mu_{B(\beta)})$ to 
$L^p(Q_{\beta}, d\mu_{\beta})$ for all $p\in [1, \infty)$.
\end{itemize}

\begin{Thm}[\cite{Arai, Arai2}]\label{BoseFKF}
Let $G_0,  \dots, G_n$ be bounded measurable functions on $\BbbR^m$. 
We set 
\begin{align}
G_j^{\rf}&=G_j\Big(\Phi_{\mathrm{S}}(f_1^{(j)}), \dots,  \Phi_{\mathrm{S}}(f_m^{(j)})\Big)
\end{align}
for $f_1^{(j)}, \dots, f_m^{(j)}\in \mathfrak{X}_{r, \gamma}^{(\beta)}$.
  For $0<t_1<t_2<\cdots<t_n<\beta$, we have
\begin{align}
&\Tr_{\Fock(\mathfrak{X})} \Big[
G_0^{\rf} e^{-t_1\dG_{{\rm s}}(A)} G_1^{\rf} e^{-(t_2-t_1)\dG_{{\rm s}}(A)}G_2^{\rf}\cdots G_n^{\rf}e^{-(\beta-t_n)\dG_{{\rm s}}(A)}
\Big]\Bigg/\mathcal{Z}_{A}(\beta) \no
=& \int_{Q_{\beta}} d\mu_{\beta} (J_0G_0^{\rf})(\Phi) (J_{t_1}G_1^{\rf})(\Phi) \cdots (J_{t_n}G_n^{\rf})(\Phi).
\end{align}
Here, note that 
\begin{align}
(J_t
G^{\rf}_j)(\Phi)&=G_j\Big(\Phi(j_t f_1^{(j)}), \dots,  \Phi(j_tf_m^{(j)})\Big).
\end{align}
\end{Thm}

\subsection{Positivity preservingness of $e^{i\Pi(f)}$}
For each $f\in \mathfrak{X}_{-1, r}^{\beta}$, we define a linear operator on $\Fock(\mathfrak{X}_{-1}^{\beta})$
by 
\begin{align}
\Pi(f)=\frac{i}{\sqrt{2}}(a(f)^*-a(f)),
\end{align}
where $\mathfrak{X}_{-1}^{\beta}$ is the complexification of $\mathfrak{X}_{-1. r}^{\beta}$.
Then $\Pi(f)$ is essentially self-adjoint. We denote its closure by the same symbol. As before, we have a natural identification $\Fock(\mathfrak{X}_{-1}^{\beta}) \cong L^2(Q_{\beta}, d\mu_{\beta})$.
Under this identification, $\Pi(f)$ can be regarded as a linear operator on $L^2(Q_{\beta}, d\mu_{\beta})$.

The following proposition will play an important role.

\begin{Prop}\label{PPPi}
For any $f\in \mathfrak{X}_{-1, r}^{\beta}$, $e^{i\Pi(f)}$ is positivity preserving, that is,
if $F\in L^2(Q_{\beta}, d\mu_{\beta})$ is a positive function, then $e^{i\Pi(f)}F$ is a positive function. 
\end{Prop}
\begin{proof}
We will apply the idea in \cite{JHB,Simon}.
First, we note the following equality:
\begin{align}
e^{i\Pi(f)}e^{i\Phi(g)}=e^{-i\la f|g\ra_{-1, \beta}} e^{i\Phi(g)}e^{i\Pi(f)}, \ g\in \mathfrak{X}_{-1, r}^{\beta}. \label{Weyl1}
\end{align}
Let $F(x_1, \dots, x_n), G(x_1, \dots, x_n)\in \mathscr{S}(\BbbR^n)$, the functions of rapid decrease.
For each $f_1, \dots, f_n, g_1, \dots, g_n\in \mathfrak{X}_{-1, r}^{\beta}$, we set
\begin{align}
\tilde{F}=F(\Phi(f_1), \dots, \Phi(f_n)),\ \ \tilde{G}=G(\Phi(f_1), \dots, \Phi(g_n)).
\end{align}
By using (\ref{Weyl1}), we have
\begin{align}
\big\la \tilde{F}|e^{i\Pi(f)}\tilde{G}\big\ra
=&e^{-\|f\|_{-1, \beta}^2/4}(2\pi)^{-n} \int_{\BbbR^{2n}} d{\bs s} d{\bs t} \hat{F}({\bs s})^* \hat{G}({\bs t})
\exp\Bigg\{ -\frac{1}{4}\bigg\|
\sum_{i=1}^n(s_i f_i-t_ig_i)
\bigg\|_{-1, \beta}^2\Bigg\}\times \no
&\times \exp\Bigg\{-\frac{i}{2} \bigg\la f\bigg|\sum_{i=1}^n(s_if_i+t_ig_i)\bigg\ra_{-1, \beta}\Bigg\}, \label{InnerFG}
\end{align}
where $\hat{f}$ indicates the Fourier transform of $f$. Let $K$ be a bounded linear operator on $L^2(\BbbR^n)$
 defined by 
 \begin{align}
 \widehat{KG}({\bs s}) =(2\pi)^{-n/2} \int_{\BbbR^n} d{\bs t}
\exp\Bigg\{ -\frac{1}{4}\bigg\|
\sum_{i=1}^n(s_i f_i-t_ig_i)
\bigg\|_{-1, \beta}^2\Bigg\}\hat{G}({\bs t}),\ \ G\in L^2(\BbbR^n).
 \end{align}
 For ${\bs c}\in \BbbR^n$, let $T_{{\bs c}} $ be the shift operator on $L^2(\BbbR^n)$: $(T_{{\bs c}}F)({\bs x})=F({\bs x}-{\bs c})$.
Then 
\begin{align}
\mbox{the RHS of (\ref{InnerFG})}=\la T_{{\bs a}}F|KT_{{\bs b}}G\ra, \label{InnerFG2}
\end{align}
where ${\bs a}=(a_1, \dots, a_n)$ with $a_i=-\la f|f_i\ra_{-1, \beta}/2$ and ${\bs b}
=(b_1, \dots, b_n)$ with $b_i=\la f|g_i\ra_{-1, \beta}/2$. 

Let $H_{{\bs s}}({\bs t})=
\exp\big\{ -\frac{1}{4}\big\|
\sum_{i=1}^n(s_i f_i-t_ig_i)
\big\|_{-1, \beta}^2\big\}
$. Because $H_{{\bs s}}$ is Gaussian, $H_{{\bs s}}$ has the  Fourier transform which is a Gaussian. In particular, 
$\check{H}_{\bs s}\ge 0$, where $\check{f}$ indicates the inverse Fourier transform of $f$.
Thus, if $G$ is positive, then $(KG)({\bs s})=\check{H}_{\bs s} * G$ is positive as well, where $*$ indicates the convolution, which implies that  the linear operator $K$ is positivity preserving.
Because the shift operator $T_{{\bs c}}$ is positivity preserving, the right hand side of (\ref{InnerFG2})
is positive.  Since any positive $\Psi\in L^2(Q_{\beta}, d\mu_{\beta})$ is a limit of such 
$F(\Phi(f_1),\dots, \Phi(f_n)),\ F\in \mathscr{S}(\BbbR^n)$ \cite[Proof of Theorem I.12]{Simon},
we conclude the assertion in Proposition \ref{PPPi}. \end{proof}

\section{ Feynman-Kac-It\^o formulas for $H_{\mathrm{rad}}$ and $H_{\mathrm{rad}, \infty}$} \label{FKIRR}
\setcounter{equation}{0}

\subsection{The Feynman-Schr\"{o}dinger representation of the radiation field}
In this study, we will employ the Feynman-Schr\"{o}dinger representation of the quantized  radiation field,
which was first introduced by Feynman \cite{Feynman}.
(Some mathematical properties of it were examined in \cite{Miyao2, Miyao3}.)
To explain the representation, we need some preliminaries:
Let $\Hf=
\sum_{\lambda=1,2} \sum_{k\in V^*} \omega(k)a(k, \lambda)^* a(k, \lambda)
$.
By (\ref{Tubutubu}), we have  
\begin{align}
\Hf=\dG_{\mathrm{s}}(\omega\oplus \omega),
\end{align}
where $\dG_{{\rm s}}(A)$ is the second quantization operator defined by (\ref{Bose2ndQ}).
Furthermore, the vector potential can  be  expressed as  
$A_j(x)=\Phi_{{\rm S}}(\eta_{x, j}^{(1)} \oplus \eta_{x, j}^{(2)})$, where
\begin{align}
\eta_{x, j}^{(\lambda)}(k)=\frac{\vepsilon_{\lambda, j}(k)}{\sqrt{\omega(k)}} \varrho(k) e^{-ik\cdot x}, \ \ \lambda=1, 2\label{DefEta}
\end{align}
with $
\varrho(k)=|\Lambda|^{-1/2} \chi_{\kappa}(k)
$.
We set 
\begin{align}
\ell_{\mathrm{even}}^2(V^*)&=\{f\in \ell^2(V^*)\, |\, f(-k)=f(k)\ \  \forall k\in V^*\},\\
\ell_{\mathrm{odd}}^2(V^*)&=\{f\in \ell^2(V^*)\, |\, f(-k)=-f(k)\ \  \forall k\in V^*\}.
\end{align}
Let us introduce subspaces of $\ell^2(V^*)$ by 
\begin{align}
\mathfrak{h}_1&=\{\vepsilon_{1, i} f\, |\, f\in \ell_{\mathrm{even}}^2(V^*),\ i=1,2,3\},\\
\mathfrak{h}_2&=\{\vepsilon_{1, i} f\, |\, f\in \ell_{\mathrm{odd}}^2(V^*),\ i=1,2,3\},\\
\mathfrak{h}_3&=\{\vepsilon_{2, i} f\, |\, f\in \ell_{\mathrm{even}}^2(V^*),\ i=1,2,3\},\\
\mathfrak{h}_4&=\{\vepsilon_{2, i} f\, |\, f\in \ell_{\mathrm{odd}}^2(V^*),\ i=1,2,3\}.
\end{align}
Because 
\begin{align}
\ell^2(V^*)=\bigcup_{i=1,2,3} \R(\vepsilon_{\lambda, i}),\ \ \ \lambda=1,2,
\end{align}
we have the following identification:
\begin{align}
\ell^2(V^*) \oplus \ell^2(V^*)=\mathfrak{h}_1\oplus\mathfrak{h}_2\oplus\mathfrak{h}_3\oplus\mathfrak{h}_4.
\label{BasicIdn}
\end{align}
Corresponding to (\ref{BasicIdn}), we have 
\begin{align}
\Fock_{\rr}&=\Fock(
 \mathfrak{h}_1\oplus \mathfrak{h}_2\oplus \mathfrak{h}_3\oplus \mathfrak{h}_4
),\\
\Hf&=\dG_{\mathrm{s}}({\bs \omega}),
\end{align}
where ${\bs \omega}$
is a self-adjoint operator defined  by $
{\bs \omega}=\omega\oplus \omega\oplus \omega\oplus \omega
$.
\begin{lemm}\label{UsefulIdnA}
There is a unitary operator $W$ satisfying the following {\rm (i)} and {\rm (ii)}:
\begin{itemize}
\item[{\rm (i)}]  $W \dG_{\mathrm{s}}({\bs \omega}) W^{-1}=\dG_{\mathrm{s}}({\bs \omega})$.
\item[{\rm (ii)}] $
WA(x)W^{-1}= \Phi_{\mathrm{S}}({\bs \theta}_x)
$, where ${\bs \theta}_x=(\theta_{1, x}, \theta_{2, x}, \theta_{3, x}, \theta_{4, x})$ with
\begin{align}
\theta_{1, x}(k)&=\frac{{\bs \varepsilon_1(k)}}{\sqrt{\omega(k)}} \varrho (k)\cos(k\cdot x), \ \ \theta_{2, x}(k)=\frac{{\bs \varepsilon_1(k)}}{\sqrt{\omega(k)}} \varrho (k)\sin(k\cdot x),\\
\theta_{3, x}(k)&=\frac{{\bs \varepsilon_2(k)}}{\sqrt{\omega(k)}} \varrho (k)\cos(k\cdot x), \ \ \theta_{4, x}(k)=\frac{{\bs \varepsilon_2(k)}}{\sqrt{\omega(k)}} \varrho (k)\sin(k\cdot x).
\end{align}
\end{itemize}
\end{lemm}
\begin{proof}Let $\Pi_{\mathrm{S}}(f)=\frac{i}{\sqrt{2}}(a(f)^*-a(f))$.
By (\ref{DefEta}), we have 
\begin{align}
A_j(x)= \Phi_{\mathrm{S}}
\Big(
\theta_{1, x, j} \oplus 0 \oplus \theta_{3, x, j} \oplus 0
\Big)+
\Pi_{\mathrm{S}}
\Big(
0\oplus \theta_{2, x, j} \oplus 0 \oplus \theta_{4, x, j}
\Big).
 \end{align} 
We set $W=e^{-i\pi N_2/2}e^{-i \pi N_4/2}$, where 
\begin{align}
N_2=\dG_{\mathrm{s}}(0\oplus 1\oplus 0\oplus 0), \ \ N_4=\dG_{\mathrm{s}}(0\oplus 0\oplus 0\oplus 1).
\end{align}
Then,  because $W 
\Pi_{\mathrm{S}}
\Big(
0\oplus \theta_{2, x, j} \oplus 0 \oplus \theta_{4, x, j}
\Big)W^{-1}=\Phi_{\mathrm{S}}
\Big(
0\oplus \theta_{2, x, j} \oplus 0 \oplus \theta_{4, x, j}
\Big)
$, we obtain (ii). To check (i) is easy.
\end{proof}

By the arguments in Section \ref{BosePre}  and Lemma \ref{UsefulIdnA}, we can regard the vector potential, $A(x)$, 
 as multiplication operators in $L^2(Q, d\mu)$. This representation is called the {\it 
 Feynman-Schr\"{o}dinger representation}, which is useful in constructing  Feynman-Kac-It\^o formulas  in the remainder of this section.

Let $\mathfrak{h}_{r, \lambda}$ be the real-Hilbert space of real-valued sequences in $\mathfrak{h}_{\lambda}$.
Let us consider the Gaussian random process  indexed by $\mathfrak{X}_{\rr, r}:=\mathfrak{h}_{r, 1}\oplus \mathfrak{h}_{r, 2}\oplus \mathfrak{h}_{r, 3}\oplus \mathfrak{h}_{r, 4}$: $\{\phi(f)\, |\, f\in \mathfrak{X}_{\rr, r}\}$, and let $(Q, \mathcal{F}, \mu)$ be its underlying probability space. The vector potential $A(x)$ can be expressed as
$\phi({\bs \theta}_x)$.  
We readily confirm that $\sum_{k\in V^*} \omega(k)^{-4}<\infty$. Hence, by choosing $A=
{\bs \omega}$, all results  in Section \ref{TraceBoseGene} hold.

\subsection{A  Feynman-Kac-It\^o formula for $H_{\mathrm{rad}}$ }

In the remainder of this section, we assume that $d\le 3$.
For notational simplicity, we express $\mathfrak{X}_{\rr, r}$ as 
$\mathfrak{X}_r$ in this section. To construct a Feynman-Kac-It\^o formula for $H_{\mathrm{rad}}$,
 we need some preliminaries.

\begin{lemm}
Let 
\begin{align}
\Theta_{ XY}=\delta_{\sigma(X)\sigma(Y)}\int_{C_{xy}}  {\bs \theta}_{ r}\cdot dr \in \mathfrak{X}_r, \ \ X=(x, \sigma(X)), Y=(y, \sigma(Y))\in \Omega, \label{LineDef}
\end{align}
where ${\bs \theta}_r$ is defined in Lemma \ref{UsefulIdnA}.
Then,  for all $s\le u \le t$, we have
\begin{align}
\int_s^t j_u\Theta(d X_u^{(j)})\in L^2\big(L_{\beta}, d\Me; \mathfrak{X}_{-1, r}^{\beta}\big), \ \ j=1, \dots, N.
\end{align}
Recall that the probability space $(L_{\beta}, \Me)$ and the random variable  $\Theta(dX_u^{(j)})$ are defined in Section \ref{FKIEl}.
\end{lemm}
\begin{proof}
Note that the line integral in (\ref{LineDef}) depends solely on the points $x$ and $y$, and thus independent
 of the path between them.  Accordingly, by choosing $C_{xy}$ as 
 $C_{xy}=\{ (1-s)x+sy\in V\, |\, s\in [0, 1]\}$, we obtain
\begin{align}
\Theta_{1, XY}(k)&=\frac{\varrho(k)}{\sqrt{\omega(k)}} {\bs \varepsilon}_{1}(k)\cdot \frac{x-y}{|x-y|}
\frac{\sin(k\cdot x)-\sin(k\cdot y)}{k\cdot (y-x)}, \no
\Theta_{2, XY}(k)&=\frac{\varrho(k)}{\sqrt{\omega(k)}} {\bs \varepsilon}_{1}(k)\cdot \frac{x-y}{|x-y|}
\frac{\cos(k\cdot x)-\cos(k\cdot y)}{k\cdot (y-x)}.
\end{align}
We can also get similar formulas for $\Theta_{3, XY}$ and $\Theta_{4, XY}$.
Using these formulas, we obtain 
\begin{align}
|
\Theta_{\lambda, XY}(k)
| \le \frac{\varrho(k)}{\sqrt{\omega(k)}}. \label{ThetaEst}
\end{align}
We have
\begin{align}
&\bigg\|
\int_s^t j_u\Theta(d X_u^{(j)})
\bigg\|^2_{
L^2\big(L_{\beta}, d\Me; \mathfrak{X}_{-1, r}^{\beta}\big)
} \no
\le &2 \int_{L_{\beta}} \sum_{i=N^{(j)}(s)+1}^{N^{(j)}(t)} \Big\la j_{u} \Theta_{X^{(j)}_{J_{i-1} J_{i}}} \Big|
j_u \Theta_{ X^{(j)}_{J_{i-1} J_{i}}}\Big\ra_{
\mathfrak{X}_{-1, r}^{\beta}
} d\Me\no
=& 2 \int_{L_{\beta}} \sum_{i=N^{(j)}(s)+1}^{N^{(j)}(t)} \Big\la  \Theta_{ X^{(j)}_{J_{i-1} J_{i}}} \Big| B(\beta)^2
 \Theta_{ X^{(j)}_{J_{i-1} J_{i}}}\Big\ra_{\mathfrak{X}_r} d\Me\no
 \le & 8\int_{L_{\beta}} \sum_{i=N^{(j)}(s)+1}^{N^{(j)}(t)} \Big\la \omega^{-1/2} \varrho \Big|B(\beta)^2 \omega^{-1/2} \varrho\Big\ra_{\ell^2_r(V^*)} d\Me\no
 =& 8 \big\|B(\beta) \omega^{-1/2} \varrho \big\|^2_{\ell^2(V^*)}\int_{L_{\beta}} N^{(j)}(t-s) d\Me. \label{L2Pf}
\end{align}
The second inequality follows from (\ref{ThetaEst}).
By Proposition \ref{NExpLe}, the right hand side of (\ref{L2Pf}) is finite. \end{proof}

\begin{lemm}\label{Cauchy}
For each  $n\in \BbbN$, set  $t_i^{(n)}=\beta i/{2^{n}},\ i=0, 1, \dots, 2^{n}$.
We define an $\mathfrak{X}_{-1, r}^{\beta}$-valued  random variable on $L_{\beta}$ by 
\begin{align}
C_n^{(j)}=\sum_{i=1}^{2^n}\int_{t_{i-1}^{(n)}}^{t_i^{(n)}} j_{t_i^{(n)}} \Theta(dX_u^{(j)}),\ \ j=1, \dots, N. \label{CnDef}
\end{align}
Then $(C_n^{(j)})_{n=1}^{\infty}$ is a Cauchy sequence in $L^2\big(L_{\beta}, d\Me; \mathfrak{X}_{-1, r}^{\beta}\big)$.
\end{lemm}
\begin{proof}
We apply the standard argument in  probability theory, see,  e.g.,  \cite{JHB,Simon2}.
First, note that 
\begin{align}
C_{n+1}^{(j)}-C_n^{(j)}=\sum_{i=1}^{2^{n}} \int^{t_{2i-1}^{(n+1)}}_{t_{2i-2}^{(n+1)}}
\Big(j_{t_{2i-1}^{(n+1)}}-j_{t_{2i}^{(n+1)}}\Big) \Theta(dX_u^{(j)}).
\end{align}
In the remainder of this proof, we abbreviate $t_i^{(n+1)}$ as $t_i$. 
For $s<t\le \beta$, we set 
\begin{align}
D_{[s, t]}=
\int_s^t \Theta (dX_u^{(j)}),\ \ j=1, \dots, N.
\end{align}
We have
\begin{align}
\Big\|C_{n+1}^{(j)}-C_n^{(j)}\Big\|^2_{
L^2\big(L_{\beta}, d\Me; \mathfrak{X}_{-1, r}^{\beta}\big)
}
\le \sum_{i=1}^{2^n}2\Big\|
\big(j_{t_{2i-1}}-j_{t_{2i}}\big)D_{
[t_{2i-2}, t_{2i-1}]
}
\Big\|^2_{
L^2\big(L_{\beta}, d\Me; \mathfrak{X}_{-1, r}^{\beta}\big)
}.
\end{align}
For each $x\ge 0$, we set 
\begin{align}
K_{\beta, n}(x)=\coth\Big(\frac{\beta x}{2}\Big)-\frac{e^{-(\beta-2^{-(n+1)})x}+e^{-2^{-(n+1)}x}}{1-e^{-\beta x}}
.
\end{align}
We readily confirm that 
\begin{align}
K_{\beta, n}(x)\le \frac{1}{2^{n+1}} R_{\beta, n}(x),\ \ R_{\beta, n}(x)=\beta x\frac{1-e^{-(\beta-s)x}}{1-e^{-\beta x}}.
\end{align}
By using this, we obtain 
\begin{align}
&\Big\|
\big(j_{t_{2i-1}}-j_{t_{2i}}\big)D_{
[t_{2i-2}, t_{2i-1}]
}
\Big\|^2_{
L^2\big(L_{\beta}, d\Me; \mathfrak{X}_{-1, r}^{\beta}\big)
}\no
=&
\int_{L_{\beta}} d\Me 
\Big\la D_{
[t_{2i-2}, t_{2i-1}]
} \Big|K_{\beta, n}(A)  D_{
[t_{2i-2}, t_{2i-1}]
}\Big\ra_{\mathfrak{X}_r}\no
\le & \frac{1}{2^{n+1}} \int_{L_{\beta}}d\Me
\Big\la D_{
[t_{2i-2}, t_{2i-1}]
} \Big|R_{\beta, n}(A)  D_{
[t_{2i-2}, t_{2i-1}]
}\Big\ra_{\mathfrak{X}_r}.
\end{align}
In addition, 
\begin{align}
&\int_{L_{\beta}}d\Me 
\Big\la D_{
[t_{2i-2}, t_{2i-1}]
} \Big|R_{\beta, n}(A)  D_{
[t_{2i-2}, t_{2i-1}]
}\Big\ra_{\mathfrak{X}_r}\no
\le & 2 \int_{L_{\beta}} d\Me \sum_{k=N^{(j)}(t_{2i-2})+1}^{N^{(j)}(t_{2i-1})} 
\Big\la \Theta_{X^{(j)}_{J_{k-1}} X_{J_k}}
\Big|R_{\beta, n}(A)   \Theta_{X^{(j)}_{J_{k-1}} X_{J_k}}
\Big\ra_{\mathfrak{X}_r}\no
\le & 8\Big\|
\omega ^{-1/2} \sqrt{R_{\beta, n}(\omega) }\varrho
\Big\|^2_{\ell^2_r(V^*)} \int_{L_{\beta}}d\Me N^{(j)}(t_{2i-1}-t_{2i-2})\no
\le& \frac{C}{2^{n+1}},
\end{align}
where $C$ is some constant independent of $n$.
In the second inequality, we used (\ref{ThetaEst}). In the last inequality, we applied Proposition \ref{NExpLe}.
To sum, we arrive at 
\begin{align}
\Big\|C_{n+1}^{(j)}-C_n^{(j)}\Big\|^2_{
L^2\big(L_{\beta}, d\Me; \mathfrak{X}_{-1, r}^{\beta}\big)
} \le \frac{C'}{2^{n+1}},
\end{align}
where $C'$ is some constant independent of $n$. Hence, for $m<n$, we have
\begin{align}
\|C_{n}-C_m\|^2_{
L^2\big(L_{\beta}, d\Me; \mathfrak{X}_{-1, r}^{\beta}\big)
}
\le \sum_{i=m}^n \frac{C'}{2^{i+1}}.
\end{align}
Thus, we are done.
 \end{proof}

\begin{define}\label{ADef}
{\rm
For each $j=1, \dots, N$, we define an $\mathfrak{X}_{-1, r}^{\beta}$-valued random variable on $L_{\beta}$
by 
\begin{align}
{\bs a}_{\beta}\big(X^{(j)}_{\bullet}\big)=\lim_{n\to \infty} C_n^{(j)},
\end{align}
where the right hand side exists in $L^2$-sense. 
}
\end{define}

\begin{Thm}\label{FKFFull}
Suppose that  $d\le 3$.
Let 
\begin{align}
Z_{\rr}(\beta)=\Tr_{\mathfrak{H}_{\rr}^{(N)}}\Big[
e^{-\beta H_{\rr}}
\Big].
\end{align}
In addition, let
\begin{align}
{\bs Z}_{\rr}(\beta)=\Tr_{\ell^2_{\mathrm{s}}(\Omega_{\neq }^N)\otimes \Fock_{\rr}}\Big[
e^{-\beta({\bs L}+\dG_{\mathrm{s}}({\bs \omega}))}
\Big]. 
\end{align}
We define a probability measure on $\mathscr{Q}_{\beta,\rr}:=L_{\beta}\times Q_{\beta,\rr}$ by 
$\mathbb{P}_{\beta, \rr}=\Me\otimes \mu_{\beta, \rr}$. Then, for all $\beta >0$,  we have 
\begin{align}
&Z_{\mathrm{rad}}(\beta)\Big/{\bs Z}_{\rr}(\beta)\no
=& \int_{\mathscr{Q}_{\beta, \rr}} d\mathbb{P}_{\beta, \rr}
\exp\Bigg\{i \Phi({\bs A}_{\beta}({\bs X}_{\bullet}))
-\int_0^{\beta} V({\bs X}_s)ds+\beta b \sum_{j=1}^N \sigma(X_0^{(j)})
\Bigg\}
(-1)^{\pi({\bs X}_{\beta})}, \label{QEDTrace}
\end{align}
where ${\bs A}_{\beta}({\bs X}_{\bullet})=\sum_{j=1}^N {\bs a}_{\beta}\big(X_{\bullet}^{(j)}\big)$.
\end{Thm}
\begin{proof}
Let   $\Theta=(\Theta_{XY})_{X, Y}$ be  a matrix defined through (\ref{LineDef}).
Note that $\Phi_{\mathrm{S}}(\Theta):=(\Phi_{\mathrm{S}}(\Theta_{XY}))_{X, Y}$ can be regarded as a magnetic potential.  
By the Trotter-Kato formula, we have 
\begin{align}
Z_{\mathrm{rad}}(\beta)\Big/{\bs Z}_{\rr}(\beta)
=\lim_{n\to \infty}\Tr\bigg[
\Big(
e^{-\beta H(\Phi_{\mathrm{S}}( \Theta))/2^n}e^{-\beta \dG({\bs \omega})/2^n}
\Big)^{2^n}
\bigg]\bigg/{\bs Z}_{\rr}(\beta).  \label{Trotter1}
\end{align}
Set $t_i=i\beta/2^n,\ i=0, 1, \dots, 2^n$. By Theorems \ref{TrFKIMHubbard} and \ref{BoseFKF}, we have
\begin{align}
&\mbox{the RHS of (\ref{Trotter1})}\no
=&\lim_{n\to \infty} \int_{\mathscr{Q}_{\beta, \rr}} d\mathbb{P}_{\beta, \rr}
\Tr_{\ell^2_{\mathrm{as}}(\Omega^N)}
\bigg[
e^{-t_1 H(\Phi_{t_1}(\Theta))} e^{-(t_2-t_1)H(\Phi_{t_2}(\Theta))} \cdots\no
&\ \ \ \cdots e^{-(t_{2^n}-t_{2^n-1}) H(\Phi_{t_{2^n}}(\Theta))}
\bigg]\bigg/\Tr_{\ell_{\mathrm{s}}^2(\Omega^N_{\neq})}\big[e^{-\beta {\bs L}}\big] \no
=&\lim_{n\to \infty} \int_{\mathscr{Q}_{\beta, \rr}} d\mathbb{P}_{\beta, \rr}
\exp\Bigg\{\sum_{j=1}^N \sum_{i=1}^{2^n} 
\mathcal{S}_{t_{i-1}, t_i}\Big(0, \Phi_{t_i}(\Theta)\, |\, X^{(j)}\Big)
-\int_0^{\beta}V({\bs X}_s)ds+\no
&\ \ \ \ \ \ +\beta b\sum_{j=1}^N \sigma(X_0^{(j)})
\Bigg\} (-1)^{\pi({\bs X}_{\beta})}. \label{Trotter2}
\end{align}
By the fact $\Phi_t(f)=\Phi(j_t f)$, we have 
\begin{align}
\sum_{j=1}^N\sum_{i=1}^{2^n} 
\mathcal{S}_{t_{i-1}, t_i}\Big(0, \Phi_{t_i}(\Theta)\, |\, X^{(j)}\Big)
= i\Phi\Bigg(\sum_{j=1}^N
\sum_{i=1}^{2^n} \int_{t_{i-1}}^{t_i} j_{t_i} \Theta(dX_u^{(j)})
\Bigg). \label{AppATro}
\end{align}
By Definition \ref{ADef}, we know that the right hand side of (\ref{AppATro})  converges to 
$i\Phi({\bs A}_{\beta}({\bs X}_{\bullet}))$ in $L^2$-sense.
Therefore, by the dominated convergence theorem, we arrive at the desired result. 
\end{proof}

\begin{rem}
{\rm 
There are some other constructions of the Feynman-Kac-It\^o formula for $H_{\rr}$,
see, e.g., \cite{FP,Miyao8}. Our construction in the present paper has the benefit of usability.
}
\end{rem}

\begin{Thm}\label{PartitionRho}
For $d\le 3$ and $\beta > 0$, there is a positive measure $\rho_{\beta, \rr}$ on $L_{\beta}$ such that 
\begin{align}
Z_{ \mathrm{rad}}(\beta) =\int_{L_{\beta}} d\rho_{\beta, \rr} (-1)^{\pi({\bs X}_{\beta})} \exp
\Bigg\{
\beta b \sum_{j=1}^N\sigma(X^{(j)}(0))
\Bigg\}. 
\end{align}
\end{Thm}
\begin{proof}
First, note that  the right hand side of (\ref{QEDTrace}) can be expressed as 
\begin{align}
\mbox{the RHS of (\ref{QEDTrace})}=\Big\la 1\Big|e^{i\Phi({\bs A}_{\beta}({\bs X}_{\bullet}))} F 1\Big\ra_{L^2(\mathscr{Q}_{\beta, \rr}, d\mathbb{P}_{\beta, \rr})}, \label{LaERa}
\end{align}
where $F=e^{-\int_0^{\beta} V({\bs X}_u)du+\beta b \sum_{j=1}^N \sigma(X^{(j)}(0))} (-1)^{\pi({\bs X}_{\beta})}$.
For each contraction operator $C$ on $\mathfrak{X}_{-1}^{\beta}$, we define  its second quantization by $\Gamma(C)1=1$ and 
$
\Gamma(C):
\Phi(f_1)\cdots \Phi(f_n)
:=:
\Phi(Cf_1)\cdots \Phi(Cf_n)
:,\ f_1, \dots, f_n\in \mathfrak{X}_{-1, r}^{\beta}
$.
Because $\Gamma(e^{i\pi/2})1=1$ and  
\begin{align}
\Gamma(e^{i\pi/2}) \Phi({\bs A}_{\beta}({\bs X}_{\bullet})) \Gamma(e^{-i\pi /2})
=\Pi({\bs A}_{\beta}({\bs X}_{\bullet})),
\end{align}
 we obtain 
\begin{align}
\mbox{the RHS of (\ref{LaERa})}&=  \Big\la 1\Big|e^{i\Pi({\bs A}_{\beta}({\bs X}_{\bullet}))} F 1\Big\ra_{L^2(\mathscr{Q}_{\beta, \rr}, d\mathbb{P}_{\beta, \rr})}\no
&=\int_{L_{\beta}}d\nu_{\beta} \int_{Q_{\beta, \rr}} d\mu_{\beta, \rr} 
e^{i\Pi({\bs A}_{\beta}({\bs X}_{\bullet}))} F({\bs X}_{\bullet}).
\end{align}
Let us  define a random variable on $L_{\beta}$ by 
\begin{align}
W({\bs X}_{\bullet})=\int_{Q_{\beta, \rr}} d\mu_{\beta, \rr} e^{i\Pi({\bs A}_{\beta}({\bs X}_{\bullet}))}.
\end{align}
By Proposition \ref{PPPi}, $W({\bs X}_{\bullet})$ is positive $\nu_{\beta}$-a.e..
By setting 
\begin{align}
d\rho_{\beta, \rr}={\bs Z}_{\rr}(\beta)W({\bs X}_{\bullet})e^{-\int_0^{\beta} V({\bs X}_u)du}d \nu_{\beta},
\end{align} we obtain the desired result. 
\end{proof}

\subsection{A  Feynman-Kac-It\^o formula for $H_{\mathrm{rad}, \infty}$ }

Let $C_n^{(j)}$ be the $\mathfrak{X}_{-1, r}^{\beta}$-valued random variable on $L_{\beta, \infty}$ defined by 
(\ref{CnDef}).  Using  arguments similar to those in  the proof of Lemma \ref{Cauchy} , we can also prove that  $
(C_n^{(j)})_{n=1}^{\infty}$ is a Cauchy  sequence in $L^2\big(L_{\beta, \infty},d\Mei; \mathfrak{X}_{-1, r}^{\beta}\big)$. Thus, we can define  an $\mathfrak{X}_{-1, r}^{\beta}$-valued random variable on $L_{\beta, \infty}$
by 
\begin{align}
{\bs a}_{\beta, \infty}\big(X^{(j)}_{\bullet}\big)=\lim_{n\to \infty} C_n^{(j)},
\end{align}
where the right hand side exists in $L^2$-sense. 
In what follows,  ${\bs a}_{\beta, \infty}$ is simply written as ${\bs a}_{\beta}$,  if no confusion arises.

Using arguments similar to those in the proof of Theorem \ref{FKFFull}, we can prove the following.
\begin{Thm}Assume that $N=|\Lambda|-1$ and $d\le 3$. 
Let $Z_{\mathrm{rad}, \infty}(\beta)$ be the partition function defined by (\ref{PartFunctNatuInf}).
Let 
\begin{align}
{\bs Z}_{\rr, \infty} (\beta)=\Tr_{P_{\mathrm{G}} \mathfrak{H}_{\rr}^{(N)}}\Big[
e^{-\beta({\bs L}_{\infty}+\dG_{\mathrm{s}}({\bs \omega}))}
\Big]. 
\end{align}
We define a probability measure on $\mathscr{Q}_{\beta, \rr, \infty}:=L_{\beta, \infty}\times Q_{\beta, \rr}$ by 
$\mathbb{P}_{\beta, \rr,  \infty}=\Mei\otimes \mu_{\beta, \rr}$. Then, for all  $\beta >0$,  we have 
\begin{align}
&Z_{\mathrm{rad}, \infty}(\beta)\Big/{\bs Z}_{\rr,  \infty}(\beta)\no
=& \int_{\mathscr{Q}_{\beta, \rr,  \infty}} d\mathbb{P}_{\beta, \rr, \infty}
\exp\Bigg\{i \Phi({\bs A}_{\beta}({\bs X}_{\bullet}))
-\int_0^{\beta} V_{{\rm o}}({\bs X}_s)ds+\beta b \sum_{j=1}^N \sigma(X_0^{(j)})
\Bigg\},
\end{align}
where ${\bs A}_{\beta}({\bs X}_{\bullet})=\sum_{j=1}^N {\bs a}_{\beta}\big(X_{\bullet}^{(j)}\big)$.
\end{Thm}

We can also prove the following assertions.

\begin{coro}Assume that  $N=|\Lambda|-1$.
For $d\le 3$ and   $\beta > 0$, we have 
\begin{align}
Z_{\mathrm{rad}, \infty}(\beta)\le Z_{\mathrm{H}, \infty}(\beta),
\end{align}
where $Z_{\mathrm{H}, \infty}(\beta)$ is given by (\ref{PartFunctHubbInfty}).
\end{coro}

\begin{Thm}\label{TRFrad}Assume that  $N=|\Lambda|-1$.
For $d\le 3$ and    $\beta > 0$, there is a positive measure $\rho_{\beta, \rr, \infty}$ on $L_{\beta, \infty}$ such that 
\begin{align}
Z_{ \mathrm{rad}, \infty}(\beta) =\int_{L_{\beta, \infty}} d\rho_{\beta, \rr,  \infty} (-1)^{\pi({\bs X}_{\beta})} \exp
\Bigg\{
\beta b \sum_{j=1}^N\sigma(X^{(j)}(0))
\Bigg\}.
\end{align}
\end{Thm}

\section{Feynman-Kac-It\^o formulas for $H_{\mathrm{HH}}$ and $H_{\mathrm{HH}, \infty}$}\label{ConstFHH}
\setcounter{equation}{0}

\subsection{The Lang-Firsov transformation}

Let us introduce a linear operator by 
\begin{align}
L=\omega^{-1}\sum_{x, y\in \Lambda} g_{xy}n_x(b_y^*-b_y).
\end{align}
$L$ is essentially anti-self-adjoint. We denote its closure by the same symbol. The unitary operator $e^L$ is called the {\it Lang-Firsov transformation} \cite{LF};
 in the study of the Holstein-Hubbard model, this transformation is often useful.
Observe that
\begin{align}
e^Lc_{x\sigma}e^{-L} =e^{i\Pi_{\mathrm{S}}(\xi)}c_{x\sigma},\ \ \ 
e^L b_xe^{-L}=b_x-\omega^{-1} \sum_{y\in \Lambda} g_{xy} n_y,
\end{align}
where $\Pi_{\mathrm{S}}(\xi)=\frac{i}{\sqrt{2}}(b(\xi)^*-b(\xi))^{**}$ and $\xi=(\xi_x)_x\in \ell^2(\Lambda)$
 with $\xi_x=\omega^{-1} \sum_{y\in \Lambda} g_{xy}$.
 Here, we used the following notation: $b(\xi)=\sum_{x\in \Lambda} \xi_x b_x$.
 Let $N_{\mathrm{p}}=\dG_{\mathrm{s}}(1)$. Using the formula
 $
 e^{-i\pi N_{\mathrm{p}}/2} b_xe^{i\pi N_{\mathrm{p}}/2}=-i b_x
 $, we arrive at the following.

 \begin{lemm}\label{LangFTr}
 Let $\mathscr{U}=e^{-i\pi N_{\mathrm{p}}/2}e^{L}$. Set $\mathbb{H}_{\mathrm{HH}}
 =\mathscr{U} H_{\hh}\mathscr{U}^{-1}
 $.
 For each $x, y\in \Lambda$, we define a vector $\zeta_{xy}\in \ell^2(\Lambda)$ by $\zeta_{xy}=\xi_x-\xi_y$.
 Then we have 
 \begin{align}
 \mathbb{H}_{\hh} &= \sum_{\sigma=\pm 1} \sum_{x, y\in \Lambda} 
 (-t_{xy})e^{i\Phi_{\mathrm{S}}(\zeta_{xy})} c_{x\sigma}^*c_{y\sigma}
 +\sum_{x\in \Lambda} U_{\mathrm{eff}, xx} n_{x, +1}n_{x, -1}+\no
 &\ \ \ +\sum_{x\neq y} U_{\mathrm{eff}, xy} n_xn_y+\omega N_{\mathrm{p}},
 \end{align}
 where $U_{\mathrm{eff}, xy}=U_{xy}-\omega^{-1}\sum_{z\in \Lambda} g_{xz} g_{zy}$ if $x\neq y$, and 
 $U_{\mathrm{eff}, xx}=U-2\omega^{-1} \sum_{z\in \Lambda} g_{xz}^2$.
 \end{lemm}

  \subsection{Trace formulas for $H_{\hh}$ and $H_{\hh, \infty}$}
  By Lemma \ref{LangFTr}, we know  that every arguments in Section \ref{FKIRR} are applicable to $\mathbb{H}_{\hh}$
  and $\mathbb{H}_{\hh, \infty}=\mathscr{U} H_{\hh, \infty} \mathscr{U}^{-1}$. 
  Below, we exhibit  trace formulas for  $H_{\hh}$ and $H_{\hh, \infty}$.

  \begin{Thm}Let
  \begin{align}
   Z_{\mathrm{HH}}(\beta)=\Tr_{\mathfrak{H}_{\hh}^{(N)}} \Big[
   e^{-\beta H_{\hh}}
   \Big].
  \end{align}
For each $d\in \BbbN$ and $\beta >0$, there is a positive measure $\rho_{\beta, \hh}$ on $L_{\beta}$ such that 
\begin{align}
Z_{ \hh}(\beta) =\int_{L_{\beta}} d\rho_{\beta, \hh} (-1)^{\pi({\bs X}_{\beta})} \exp
\Bigg\{
\beta b \sum_{j=1}^N\sigma(X^{(j)}(0))
\Bigg\}. 
\end{align}
\end{Thm}

  \begin{Thm}\label{TRFrad}
  Assume that $N=|\Lambda|-1$.
  Let $Z_{\mathrm{HH}, \infty}(\beta)$ be the partition function defined by (\ref{PartFunctNatuInf}).
For each $d\in \BbbN$ and $\beta >0$, there is a positive measure $\rho_{\beta, \hh, \infty}$ on $L_{\beta, \infty}$ such that 
\begin{align}
Z_{ \hh, \infty}(\beta) =\int_{L_{\beta, \infty}} d\rho_{\beta, \hh,  \infty} (-1)^{\pi({\bs X}_{\beta})} \exp
\Bigg\{
\beta b \sum_{j=1}^N\sigma(X^{(j)}(0))
\Bigg\}.
\end{align}
\end{Thm}

We can also prove the following proposition.

\begin{Prop}  Assume that $N=|\Lambda|-1$.
For all $d\in \BbbN$ and  $\beta > 0$, we have 
\begin{align}
Z_{\hh, \infty}(\beta)\le Z_{\mathrm{H}, \infty}(\beta).
\end{align}
\end{Prop}

\section{Random loops representations}\label{ConstLP}
\setcounter{equation}{0}

In this section, we derive random loop representations for $Z_{\natural}(\beta)$ and $Z_{\natural, \infty}(\beta),\ \natural=\rr, \hh$.
Similar representations are  known to be a useful tool in quantum spin systems, see, e.g., \cite{AN, Uel}.
As we will see in Sections  \ref{SecPf2} and \ref{UseForU}, the representations also play important roles  in the proof of Theorem  \ref{FinNTEnv}.

For  each ${\bs m} \in L_{\beta}$,   the corresponding path $({\bs X}_t({\bs m}))_{0\le t \le \beta}$  satisfies the following properties:
\begin{itemize}
\item[(i)] $\sigma^{(j)}_t({\bs m})$ is constant in $t$ for all $j=1, \dots, N$;
\item[(ii)] there exists a dynamically allowed  permutation $\tau$ such that ${\bs X}_{\beta}({\bs m})=\tau{\bs X}_0({\bs m})$;
\item[(iii)] ${\bs X}_t({\bs m})\in \Omega_{\neq}^N$ for all $t\in [0, \beta]$, that is, there are no encounters of electrons of equal spin.
\end{itemize}
Furthermore, the loops are associated with the  path $({\bs X}_t({\bs m}))_{0\le t \le \beta}$;
the loops are obtained by the following manner\cite{AL}:
\begin{itemize}
\item we start drawing the loops from the  $t=0$ location of any of the electrons;
\item we trace the electron's location forward in space-time  until its first encounter with another electron;
\item at the encounter point, the tracing line switches to the world line of the other electron in the {\it reversed} orientation in time;
\item such an  orientation switch is repeated whenever an electron encounter is reached;
\item when trace line reaches the time at $t=0$ or $t=\beta$, it reemerges at the same location with time treated as periodic;
\item the above procedures are continued until a trace line is closed.
\end{itemize}
For readers\rq{} convenience, we give an example of the  loop associated with a path in Figures \ref{Fig1} and \ref{CorrLoop}.

\begin{figure}[ht]
    \begin{tabular}{cc}
     
      \begin{minipage}[t]{0.5\hsize}
        \centering
        \input{fig1.tex}
        \caption{
        The collection of electron world lines. Electron $1$ and  electron  $3$ have the spin value  $+1$;  electron $2$
        has the  spin value  $-1$.
        }
        \label{Fig1}
      \end{minipage} &
     
      \begin{minipage}[t]{0.45\hsize}
        \centering
        \input{Loop2.tex}
        \caption{The loop corresponding to Figure \ref{Fig1}.
        The black colored loop has the winding number $1$ and the parity $+1$.
        }
        \label{CorrLoop}
      \end{minipage}
     
    \end{tabular}
  \end{figure}

For each ${\bs m}\in L_{\beta}$, let $\varGamma({\bs m})$ be the collection of all loops associated with the path $({\bs X}_t({\bs m}))_{0\le t\le \beta}$. Each loop in $\varGamma({\bs m})$ is a closed  trajectory $\gamma:
 [0, \ell\beta]_{\mathrm{P}} \to \Lambda\times [0, \beta]_{\mathrm{P}}$, where  $[0, \beta]_{\mathrm{P}}$ is the interval $[0, \beta]$ with periodic boundary conditions, i.e., the torus of length $\beta$.
 For each $\gamma\in \varGamma({\bs m})$, the {\it winding number} is defined by the following line integral:
 \begin{align}
 w(\gamma)=\bigg|
 \int_{\gamma} \frac{d\tau}{\beta}
 \bigg|. \label{DefWinding}
 \end{align}
 Note that, because $\gamma$ is a piecewise smooth curve,
 the right hand side of (\ref{DefWinding}) is well defined. 
 
 Our main result in this section is stated as follows.

 \begin{Thm}[Random loop representation I]\label{RLR}
 Suppose that $d\le 3$ for $\natural=\mathrm{rad}$, $d\in \BbbN$ for $\natural=\mathrm{HH}$.
 For $\natural=\rr, \hh$, one obtains 
 \begin{align}
 Z_{\natural}(\beta)=\int_{L_{\beta}} d\rho_{\beta, \natural}({\bs m}) (-1)^{\pi({\bs X}_{\beta} ({\bs m}))} \prod_{\gamma\in \varGamma({\bs m})} 
 \cosh \Big(
 \beta b w(\gamma)
 \Big).
 \end{align}
 \end{Thm}

 For ${\bs m}\in L_{\beta, \infty}$,
   we can also associate the path $({\bs X}_t({\bs m}))_{0\le t\le \beta}$ with the loops in the   same manner as  above.   Then, we can  prove the following.
 
  \begin{Thm}[Random loop representation II]\label{RLR2}
  Suppose that $d\le 3$ for $\natural=\mathrm{rad}$, $d\in \BbbN$ for $\natural=\mathrm{HH}$.
 For $\natural=\rr, \hh$, one obtains 
 \begin{align}
 Z_{\natural, \infty}(\beta)=\int_{L_{\beta, \infty}} d\rho_{\beta, \natural, \infty}({\bs m})  \prod_{\gamma\in \varGamma({\bs m})} 
 \cosh \Big(
 \beta b w(\gamma)
 \Big).
 \end{align}
 \end{Thm}

 \subsection{Proof of Theorem \ref{RLR}}
 For each loop $\gamma\in \varGamma({\bs m})$, set $\gamma_{t}=\gamma\cap {\bs X}_t({\bs m})$,  the  cross section of $\gamma$ with the cutting plane described   by the equation $\tau=t$ in 
 the $\tau$-$X$ plane. 
 For notational simplicity, suppose that $\gamma_t=(X_{t}^{(i_1)}({\bs m}),\dots, X_t^{(i_n)}({\bs m}))$
 with $X_t^{(i)}({\bs m})=(x^{(i)}_t({\bs m}), \sigma_0^{(i)}({\bs m}))$. Note that $i_1, \dots, i_n$ and $n$ could depend on $t$.
 Then we readily  confirm that   
 \begin{align}
 \sigma_t^{(i_1)}({\bs m})+\cdots+\sigma_t^{(i_n)}(\bs m)=\vepsilon(\gamma) w(\gamma), \ \ 
 \vepsilon(\gamma)=\pm 1, \label{base}
 \end{align}
 where
 we understand that the left hand side of (\ref{base}) equals $0$, provided that  $\gamma_t=\varnothing$.
 The factor $\vepsilon(\gamma)$ is called the  {\it parity} of $\gamma$, see Figures \ref{Fig1} and \ref{CorrLoop}.
 
 The following lemma is an immediate consequence of (\ref{base}).
 \begin{lemm}\label{SWinding}
 For each ${\bs m} \in L_{\beta}$, we have 
 \begin{align}
 \sum_{j=1}^N\sigma_0^{(j)}({\bs m})=\sum_{\gamma\in \varGamma({\bs m})} \vepsilon(\gamma)w(\gamma).
 \end{align}
 \end{lemm}

 Let $\big(X_t^{(j)}({\bs m})\big)_{0\le t \le \beta}$ be  a  trajectory of the $j$-th electron with 
 \begin{align}
 X_t^{(j)}({\bs m})=\big(x_t^{(j)}({\bs m}), \sigma_0^{(j)}(\bs m)\big).
 \end{align}
 For each $t\in [0, \beta]$, we set 
 $
 \overline{X}_t^{(j)}({\bs m})=\big(x_t^{(j)}({\bs m}), -\sigma_0^{(j)}(\bs m)\big)
 $,  the  spin-reversed point corresponding to  $X_t^{(j)}({\bs m})$ in space-time.
 For $\gamma\in \varGamma({\bs m})$ characterized by 
 \begin{align}
 \gamma_t=\big(
 X_t^{(i_1)}(\bs m), \dots, X_t^{(i_n)}({\bs m})
 \big),
 \end{align}
  the {\it conjugate loop}  $\overline{\gamma}$ is defined through the relation
 $\overline{\gamma}_t=\big(
 \overline{X}_t^{(i_1)}({\bs m}), \dots, \overline{X}_t^{(i_n)}({\bs m})
 \big)$ for all $ 0\le t \le \beta$.
 
 Let  ${\bs m} \in L_{\beta}$.  We express $\varGamma({\bs m})$  as $\varGamma({\bs m})
 =\{\gamma_1, \dots, \gamma_K\}$, provided that $\varGamma({\bs m})\neq \varnothing$.
 As before, we set 
 $\gamma_{\alpha, t}=\gamma_{\alpha}\cap {\bs X}_t({\bs m})=
 \big(
 X_{\alpha, t}^{(i_1)}({\bs m}), \dots, X_{\alpha, t}^{(i_n)}({\bs m})
 \big)$ with $X_{\alpha,  t}^{(i)}({\bs m})=\big(
 x_{\alpha, t}^{(i)}({\bs m}), \sigma_{\alpha, 0}^{(i)}({\bs m})
 \big)$. For each $j\in \{1, \dots, N\}$, there exists a $k(j)\in \{1, \dots, K\}$ such that 
 $X_{t=0}^{(j)}({\bs m})\in \gamma_{k(j)}$. With this notation, we define a bijective map 
 $g_j: L_{\beta}\to L_{\beta}$ through the following relation:
 \begin{align}
 \varGamma(g_j{\bs m})=\big\{\gamma_1^{\flat_j}, \dots,  \gamma_K^{\flat_j}\big\}, \label{GammagEx}
 \end{align}
 where $\gamma_{\alpha}^{\flat_j}$ is given by the following manner:
 \begin{itemize}
 \item[(a)] $\gamma_{k(j)}^{\flat_j}=\overline{\gamma}_{k(j)}$, 
 \item[(b)] remainder loops $\{\gamma_{\alpha}^{\flat_j}\}_{\alpha\neq k(j)}$ are uniquely determined by the rearrangement of the spin configuration of electron trajectories needed to maintain (i)-(iii).
 \end{itemize}
 Roughly speaking, (a) the map $g_j$ flips the
spin values  along the loop  containing the $t =0$ position of the $j$-th particle; (b) the spin flip induces the rearrangement of  the spin configuration of remainder electron trajectories in order to maintain (i)-(iii), see Example \ref{ExReverse}.

\begin{example}\label{ExReverse}{\rm 
For reader's convenience, let us consider a path give in Figure \ref{A}.
 \begin{figure}[H]
    \begin{tabular}{cc}
      \begin{minipage}[t]{0.5\hsize}
        \centering
        \input{fig4.tex}
        \caption{
        }
        \label{A}
      \end{minipage} &     
      
      \begin{minipage}[t]{0.45\hsize}
        \centering
        \input{fig30.tex}
        \caption{
        }
        \label{B}
      \end{minipage}  
      \end{tabular}
      \end{figure}
 As shown in Figure \ref{B}, this path has two loops, i.e.,  the black colored loop and the green colored loop.
 In this case, the action of $g_2$ induces the change of the spin configurations from Figure \ref{C} to Figure \ref{D}.
 \begin{figure}[H]
     \begin{tabular}{cc}
      \begin{minipage}[t]{0.5\hsize}
        \centering
        \input{fig4.tex}
        \caption{
        }
        \label{C}
      \end{minipage} &     
      
      \begin{minipage}[t]{0.45\hsize}
        \centering
        \input{fig5.tex}
        \caption{
        }
        \label{D}
      \end{minipage}   
    \end{tabular}
  \end{figure}
The action of $g_1$ to this path  induces the same change of the spin configuration.
The action of $g_3$ to this path only flips  the spin value  along  the trajectory of electron 3. $\Box$ 
}
 \end{example}

 The composition map  $g_i \circ g_j$ will be simply  denoted by $g_ig_j$.
 Trivially, we have the following:
 \begin{itemize}
 \item $g_jg_j=\mathrm{Id}$, the identity map on $L_{\beta}$,  for all $j\in \{1, \dots, N\}$;
 \item $g_ig_j=g_j g_i$ for all $i, j=\{1, \dots, N\}$.
 \end{itemize}
 Let $\mathcal{G}$ be the symmetry group generated by $g_1, \dots, g_N$.
 Each $g\in \mathcal{G}$ can be expressed as $g=g^{{\bs \xi}}$, where
 \begin{align}
 g^{{\bs \xi}}=g_1^{\xi_1}g_2^{\xi_2}\cdots g_N^{\xi_N},\ \ {\bs \xi}=(\xi_1, \dots, \xi_N) \in \{0, 1\}^N.
 \end{align}
 Here, we understand that $g_j^0=\mathrm{Id}$ and $g_j^1=g_j$.
 In particular, $\mathcal{G}$ has $2^N$ elements.

 \begin{lemm}\label{Inv}
 We have the following:
 \begin{itemize}
 \item[{\rm (i)}] $\rho_{\beta, \natural}(g E)=\rho_{\beta, \natural}(E)$ for each  $E\in \mathcal{F}_{\beta}^{N}$,  $g\in \mathcal{G}$ and $\natural=\rr, \hh$, where $\mathcal{F}_{\beta}^N:=\mathcal{F}^N \cap L_{\beta}$. Here, the reader should recall that $\mathcal{F}$ is the $\sigma$-algebra on $M$ introduced in Section \ref{SingleElFKF}. 
 \item[{\rm (ii)}] $
 (-1)^{\pi({\bs X}_{\beta}( g{\bs m}))}=
 (-1)^{\pi({\bs X}_{\beta}( {\bs m}))}
 $ for all $g\in \mathcal{G}$.
 \end{itemize}
 \end{lemm}
 \begin{proof}
 (i) First, note that, by the arguments in the proof of Theorem \ref{PartitionRho}, we have
 \begin{align}
 \rho_{\beta, \mathrm{rad}}(E)={\bs Z}_{\rr}(\beta) \int_E d\nu_{\beta} W({\bs X}_{\bullet}).
  \end{align} 
  For all $g\in \mathcal{G}$ and ${\bs m} \in L_{\beta}$, we readily confirm that
  \begin{align}
  W({\bs X}_{\bullet}(g{\bs m}))&=W({\bs X}_{\bullet}({\bs m})),\\
  \int_0^{\beta} V({\bs X}_u(g{\bs m}))du&=\int_0^{\beta} V({\bs X}_u({\bs m}))du,\\
  \nu_{\beta}(gE) &= \nu_{\beta}(E).
   \end{align} 
 Taking these properties into account, we can show (i). 
 
 To check (ii) is easy. 
 \end{proof}

 Let ${\bs m}\in L_{\beta}$. Corresponding to $\varGamma({\bs m})=\{\gamma_1, \dots, \gamma_K\}$, there
 is a unique partition $(I_{\gamma})_{\gamma\in \varGamma({\bs m})}$ of $\{1, \dots, N\}$ such that 
 $
 X_{t=0}^{(i)}({\bs m}) \in \gamma
 $ for all $i\in I_{\gamma}$. 
 For each ${\bs \xi}=(\xi_i)_{i=1}^N\in \{0, 1\}^N$,
 we set ${\bs \xi}_{\gamma}=(\xi_i)_{i\in I_{\gamma}} \in \{0, 1\}^{|I_{\gamma}|}$. Trivially, 
 ${\bs \xi}=\bigcup_{\gamma\in \varGamma({\bs m})} {\bs \xi}_{\gamma}$. With this notation, we have the 
 following decomposition:
 $
 g^{{\bs \xi}}=\prod _{\gamma\in \varGamma({\bs m})} g^{{\bs \xi}_{\gamma}},
 $
 where $
 g^{{\bs \xi}_{\gamma}}=\prod_{i\in I_{\gamma}} g_i^{{\bs \xi}_i}
 $.
 
 \begin{lemm} \label{BBxi}
 For each  ${\bs m}\in L_{\beta}$ and ${\bs \xi}\in \{0, 1\}^N$, we define  random variables on $L_{\beta}$ by
 \begin{align}
 B({\bs m})&=\beta b \sum_{\gamma\in \varGamma({\bs m})} \vepsilon(\gamma) w(\gamma),\\
 B_{{\bs \xi}}({\bs m}) &= \beta b \sum_{\gamma\in \varGamma({\bs m})} \vepsilon(\gamma)
 (-1)^{\sum_{i\in I_{\gamma}} \xi_i} w(\gamma).
 \end{align}
 Then we have
 \begin{align}
 B(g^{\bs \xi} {\bs m})=B_{\bs \xi}({\bs m}). \label{ActgB}
 \end{align}
 \end{lemm}
 \begin{proof}
 We will provide a sketch of the proof.
 Fix $j\in \{1, \dots, N\}$ arbitrarily. We continue to use  notation (\ref{GammagEx}).
 By the definition of $g_j$, we see that 
 \begin{align}
 B(g_j{\bs m})=\beta b \sum_{i=1}^K \vepsilon(\gamma_i) (-1)^{\delta_{i, k(j)}} w(\gamma_i)=B_{{\bs \xi}_j}({\bs m}),
 \end{align}
 where ${\bs \xi}_j=(\delta_{i, k(j)})_{i=1}^N$. Hence, we obtain (\ref{ActgB}) when ${\bs \xi}={\bs \xi}_j$. To
 extend this argument to general ${\bs \xi}$ is not so hard. 
 \end{proof}
 
 \subsubsection*
 {\it Proof of Theorem \ref{RLR}}

 By Lemmas \ref{SWinding}, \ref{Inv} and \ref{BBxi}, we observe that 
 \begin{align}
 Z_{\natural}(\beta)&=\int_{L_{\beta}} d\rho_{\beta, \natural}(g^{\bs \xi}{\bs m})  (-1)^{\pi({\bs X}_{\beta}({\bs m}))} e^{B({\bs m})}\no
 &=\int_{L_{\beta}} d\rho_{\beta, \natural}({\bs m})(-1)^{\pi({\bs X}_{\beta}({\bs m}))} e^{B_{\bs \xi}({\bs m})}
 \end{align}
 for all ${\bs \xi}\in \{0, 1\}^N$ and $\natural=\rr, \hh$.
 Consequently, we obtain
 \begin{align}
 &Z_{\natural}(\beta)\no
 &= \sum_{{\bs \xi}\in \{0, 1\}^N} \frac{1}{2^N} \int_{L_{\beta}}d\rho_{\beta, \natural}({\bs m})(-1)^{\pi({\bs X}_{\beta}({\bs m}))}
 e^{B_{\bs \xi}({\bs m})}\no
 &= \int_{L_{\beta}}d\rho_{\beta, \natural}({\bs m}) (-1)^{\pi({\bs X}_{\beta}({\bs m}))}
 \Bigg[
 \prod_{\gamma\in \varGamma({\bs m})}\frac{1}{2^{|I_{\gamma}|}}
 \sum_{{\bs \xi}_{\gamma}\in \{0, 1\}^{|I_{\gamma}|}} 
\exp\Big\{\beta b \vepsilon(\gamma) (-1)^{\sum_{i\in I_{\gamma}} \xi_i} w(\gamma)
\Big\} \Bigg]
 \no
 &=\int_{L_{\beta}}d\rho_{\beta, \natural}({\bs m})  (-1)^{\pi({\bs X}_{\beta}({\bs m}))}
 \prod_{\gamma\in \varGamma({\bs m})} \cosh\big(
 \beta w(\gamma)
 \big).
 \end{align}
 Thus, we are done. \qed

 \subsection{Proof of Theorem \ref{RLR2}}
 Even in the case where $U=\infty$, Lemmas \ref{SWinding}, \ref{Inv} and \ref{BBxi} hold true. Therefore, by using arguments similar to those in the proof of Theorem \ref{RLR}, we can prove Theorem \ref{RLR2}. \qed

\section{Proof of Theorem \ref{FinNTEnv}}\label{SecPf2}
\subsection{A useful expression of $Z_{\natural, \infty}(\beta)$}
\setcounter{equation}{0}

For each $\tau\in \mathfrak{S}_N$, we set 
\begin{align}
L_{\beta, \infty}(\tau)=\{{\bs m}\in L_{\beta, \infty}\, |\, {\bs X}_{\beta}({\bs m})=\tau {\bs X}_0({\bs m})\}.
\end{align}
We say that $\tau\in \mathfrak{S}_N$ is an {\it allowed permutation}, if $\rho_{\beta, \natural}(L_{\beta, \infty}(\tau))\neq 0$.
Using Theorem \ref{RLR2} and the fact  $L_{\beta, \infty}=\bigcup_{\tau\in \mathfrak{S}_N} L_{\beta, \infty}(\tau)$,
we obtain 
\begin{align}
Z_{\natural, \infty}(\beta)=\sum_{\tau\in {\mathfrak{S}_N}} \mathcal{V}_{\natural, \beta}(\tau), \label{DecompZ}
\end{align}
where 
\begin{align}
\mathcal{V}_{\natural, \beta}(\tau)=\int_{L_{\beta, \infty}(\tau)}d\rho_{\beta, \natural}({\bs m}) 
 \prod_{\gamma\in \varGamma({\bs m})} \cosh\big(
 \beta w(\gamma)
 \big).
\end{align}
Here, if $\tau$ is not an allowed permutation, we set $
\mathcal{V}_{\natural, \beta}(\tau)=0
$.

To prove Theorem \ref{FinNTEnv}, we need the following theorem.

\begin{Thm}\label{FiniteNT} 
Assume  that $N=|\Lambda|-1$. 
Let $\mathscr{P}_N$ be the set of all partitions of $N$.
Suppose that $0<\beta<\infty$ and $0<b$. 
Suppose that $d\le 3$ for $\natural=\mathrm{rad}$,  $d\in \BbbN$ for $\natural=\mathrm{HH}$.
For $\natural =\mathrm{rad}, \mathrm{HH}$ and ${\bs n} =\{n_1, \dots, n_{\ell}\}\in \mathscr{P}_N$, there is a positive function $D_{{\bs n}, \natural }(\beta)$ independent of $b$ such that 
\begin{itemize}
\item[{\rm (i)}] $\displaystyle 
Z_{\natural , \infty}(\beta)=\sum_{{\bs n} \in \mathscr{P}_N} D_{{\bs n}, \natural }(\beta) \prod_{i=1}^{\ell} \cosh(\beta b n_i)
$,
\item[{\rm (ii)}] $D_{{\bs n}, \natural }(\beta)$ is strictly positive for all $0<\beta$ if,  and only if, 
there is an allowed permutation whose cyclic lengths are $n_1, \dots, n_{\ell}$.
\end{itemize}

\end{Thm}
\begin{proof}
  Suppose that $\tau$ is an allowed permutation.
Let us  write down  $\tau$  in cycle notation as 
\begin{align}
\tau=\vepsilon_1\cdots \vepsilon_{\ell}, \label{CycleN}
\end{align}
 where $
\vepsilon_i$ is  an $n_i$-cycle: $\vepsilon_i=(k_1 \cdots k_{n_i})$ with $k_{j}\in \{1, \dots, N\}$.
Without loss of generality, we may assume that   $(\vepsilon_i)_{i=1}^{\ell}$ satisfies 
$\sum_{i=1}^{\ell} n_i=N$.\footnote{For example, let us consider a permutation 
$\tau=
\begin{pmatrix}
1&2&3&4&5&6&7&8&9\\
2&3&1&7&6&8&9&5&7
\end{pmatrix}
$. In this case, we can express $\tau$ as 
$\tau=(1\ 2\ 3) (4\ 7\ 9)(5\ 6\ 8)$. Thus, $\sum_{i=1}^3n_i=3+3+3=9$ holds.} 
In this way,   $\tau$ naturally corresponds to                                                                                                                                                                                                                                                                                                                                                                                                                                                                                                                                                                                                                                                                                                                                                                                                                                                                                                                                                                                                                                                                                                                                                                                                                                                                                                                                                                                                                                                                                                                                                                                                                                                                                                                                                                                                                                                                                                                                                                                                                                                                                                                                                                                                                                                                                                                                                                                                                                                                                                                                                                                                                             $\{n_1, \dots, n_{\ell}\}\in \mathscr{P}_N$. We denote by ${\bs n}(\tau)$ the partition of $N$ with respect to $\tau$.
  We can associate  expression (\ref{CycleN}) with the  set of loops $\varGamma({\bs m})
  =\{\gamma_1, \dots, \gamma_{\ell}\}$ such that 
  each $\gamma_i$ has the winding number $n_i$.
  Hence, we obtain
  \begin{align}
  \mathcal{V}_{\natural, \beta}(\tau)= \rho_{\beta, \natural}(L_{\beta, \infty}(\tau))(\beta) \prod_{i=1}^{\ell} \cosh(\beta b n_i).
  \end{align}
  Hence, by choosing
  \begin{align}
   D_{{\bs n}, \natural }(\beta) =\sum_{\tau: {\bs n}(\tau)={\bs n}}\rho_{\beta, \natural}(L_{\beta, \infty}(\tau)),
  \end{align}
   we conclude the assertions in Theorem \ref{FiniteNT}. 
  \end{proof}

  \subsection{Proof of Theorem \ref{FinNTEnv}}
  Let $\mathbb{E}$ be  an expectation over permutations defined by 
  \begin{align}
  \mathbb{E}[F({\bs n})]=\sum_{{\bs n}\in \mathscr{P}_N}  D_{{\bs n}, \natural }(\beta) \prod_i \cosh(\beta b n_i)
  F({\bs n})\Big/ Z_{\natural,\infty}(\beta).
  \end{align}
  Then we find that 
  \begin{align}
  \la S_{\mathrm{tot}}^{(3)}\ra_{\natural, \infty}(b; \beta)=&\frac{1}{2\beta} \frac{\partial}{\partial b} \log Z_{\natural, \infty}(\beta)\no
  =& \frac{1}{2}\mathbb{E} \Big[
  \sum_i n_i \tanh (\beta b n_i)
  \Big]. \label{ExPart}
  \end{align}
  Because $\tanh(\beta b n)>\tanh (\beta b)$ for all $n\ge 2$, we find that 
  $
  \mbox{RHS of (\ref{ExPart})}>\frac{N}{2} \tanh (\beta b)
  $. \qed

\setcounter{equation}{0}
\section{A useful expression for $Z_{\natural}(\beta)$} \label{UseForU}
By using features of the one-dimensional system, we get the following expression for $Z_{\natural}(\beta)$.
Note that the theorem below could be  useful when we examine finite-temperature extensions of the Lieb-Mattis theorem \cite{LM}
 for  $H_{\hh}$ and $H_{\rr}$. (Remark that a simple extension of \cite{LM}  to  $H_{\hh}$ can be found in \cite{Miyao6}.)
\begin{Thm}\label{AZ1}
Let us consider a  one-dimensional system.
For each $m\in \mathrm{spec}(S_{\mathrm{tot}}^{(3)})$, let  
\begin{align}
Z_{\natural, M}(\beta; m)=\Tr_{\mathfrak{H}_{\natural, M}^{(N)}[m]} \Big[
e^{-\beta H_{\natural}}
\Big],\ \natural=\rr, \hh, 
\end{align}
where $\mathfrak{H}_{\natural, M}^{(N)}[m]$ is the $m$-subspace defined by 
$\mathfrak{H}_{\natural, M}^{(N)}[m]= \mathfrak{H}_{\rh, M}^{(N)}[m]\otimes \Fock_{\natural}$ with $\mathfrak{H}_{\rh, M}^{(N)}[m]=
\ker(S_{\mathrm{tot}}^{(3)}-m)
$.
We set $Y_M^{(k)}(m)=\dim \mathfrak{H}_{\rh, M}^{(k)}[m]$.
For all $\beta>0, \ b>0$ and  $
\natural =\mbox{{\rm rad},  {\rm HH}}
$, there are  functions 
$C_{N, \natural }(\beta), C_{N-2, \natural }(\beta), \dots, C_{0, \natural }(\beta)$ or $C_{1, \natural }(\beta)$ which are independent of $m$ such that  
\begin{align}
Z_{\natural }(\beta)=\sum_{k=0}^N C_{k, \natural }(\beta) \{\cosh(\beta b)\}^k, \label{LM1}
\end{align}
and 
\begin{align}
Z_{\natural , M}(\beta; m)=\sum_{k=0}^N C_{k, \natural }(\beta) Y_M^{(k)}(m), \label{LM2}
\end{align}
where we understand that $C_{k, \natural }(\beta)\equiv 0$ if $N-k$ is odd, and $Y_M^{(k)}(m)\equiv 0$ if $k<2|m|$. 
\end{Thm}
\begin{proof}
In case of the one-dimensional chain, $w(\gamma)$ can take the values $0$ and $1$.
Taking this fact into consideration, we set
\begin{align}
L_{\beta}(k)=\Big\{
{\bs m} \in L_{\beta}\, \Big| \mbox{
the number of loops in $\varGamma({\bs m})$ with $w(\gamma)=1$ is equal to $k$
}
\Big\}.
\end{align}
Because $L_{\beta}=\bigsqcup_{k=0}^NL_{\beta}(k)$, we have
\begin{align}
Z_{\natural}(\beta)=\sum_{k=0}^NC_{k, \natural}(\beta) \{\cosh(\beta b)\}^k,\ \ 
C_{k, \natural}(\beta)=\int_{L_{\beta}(k)} d\rho_{\beta, \natural} (-1)^{\pi({\bs X}_{\beta})}.
\label{ZNaturalC}
\end{align}

Next, we will prove (\ref{LM2}). First, note that 
\begin{align}
Z_{\natural}(\beta)=\sum_{m=-N/2}^{N/2} Z_{\natural, M}(\beta; m).
\end{align}
Let $z=e^{\beta b}$. By Theorem \ref{TRFrad}, we can express $Z_{\natural, M}(\beta; m)$ as 
\begin{align}
Z_{\natural, M}(\beta, m)=z^m\int_{
\big\{
\sum_{j=1}^N\sigma(X^{(j)}_0)=2m
\big\}
}d\rho_{\beta, \natural}
(-1)^{{\bs X}_{\beta}}. \label{ZZm}
\end{align}
In addition, by (\ref{ZNaturalC}), we have
\begin{align}
Z_{\natural}(\beta)=\sum_{k=0}^N C_{k, \natural}(\beta) \bigg(
\frac{z+z^{-1}}{2}
\bigg)^k.\label{ZZm2}
\end{align}
Comparing (\ref{ZZm}) and (\ref{ZZm2}), we obtain (\ref{LM2}). \end{proof}

\appendix

\section{A useful proposition} \label{AppAUse}
\setcounter{equation}{0}

The following proposition is needed in Section \ref{FKIRR}.

\begin{Prop}\label{NExpLe}
For all $j=1, \dots, N$, 
we have the following:
\begin{itemize}
\item[{\rm (i)}] $\displaystyle 
\int_{L_{\beta} }d \Me N^{(j)}(t)=O(t)
$ as  $t\to +0$.
\item[{\rm (ii)}]
 $\displaystyle 
\int_{L_{\beta, \infty} }d \Mei N^{(j)}(t)=O(t)
$ as  $t\to +0$.
\end{itemize}
\end{Prop}
\begin{proof}
(i) 
First, we note  that, by (\ref{Measurenu}),
\begin{align}
\int_{L_{\beta}} d\Me F(-1)^{\pi({\bs X}_{\beta})}=\sum_{{\bs X}\in [\Omega_{\neq}^N]} 
\sum_{\tau\in \mathfrak{S}_N({\bs X})}\frac{\mathrm{sgn}(\tau)}{N!}
\mathrm{E}_{{\bs X}}\Big[F 1_{\{{\bs X}_{\beta}=\tau {\bs X}\}\cap D}\Big] \Bigg/
\Tr_{\ell^2_{\mathrm{s}}(\Omega^N_{\neq})}
\Big[
e^{-\beta {\bs L}}
\Big].
\end{align}
Thus, it suffices to show that $\mathrm{E}_{{\bs X}}[N^{(j)}(t)]=O(t)$ as $t\to +0$.
By Lemma \ref{SnDistribution}, we find that 
\begin{align}
\mathrm{E}_{{\bs X}}[N^{(j)}(t)]&=\sum_{n=0}^{\infty}nP_{X^{(j)}}\Big(N^{(j)}(t)=n\Big)\no
 &=\sum_{n=0}^{\infty} nP_{X^{(j)}}\Big(J_n\le t<J_{n+1}\Big)\no
 &=\sum_{n=0}^{\infty} nP_{X^{(j)}}\Big(S_1+\cdots+S_n\le t<S_1+\cdots+S_{n+1}\Big)\no
 & \le d_0t-1+e^{-d_0t}\no
 &=O(t)
\end{align}
as $t\to +0$.

(ii) To prove (ii), we remark that 
\begin{align}
\int_{L_{\beta, \infty}} d\Mei F=
\sum_{{\bs X}\in [\Omega_{\neq, \infty}^N]} \sum_{{\tau\in \mathfrak{S}_{N, \infty}({\bs X})}}
\frac{1}{N!}\mathrm{E}_{\bs X}\Big[F 1_{\{{\bs X}_{\beta}=\tau {\bs X}\} \cap D_{\beta}}\Big] \Bigg/
\Tr_{P_{\mathrm{G}}\ell^2_{\mathrm{s}}(\Omega^N_{\neq})}
\Big[
e^{-\beta T}
\Big].
\end{align}
Therefore, it suffices to prove that $\mathrm{E}_{{\bs X}}[N^{(j)}(t)]=O(t)$ as $t\to +0$.
But this has been already proved in the above.  \end{proof}

\begin{lemm}\label{SnDistribution}
Let $d_0=d(0)$. (Recall that $d(x)$ is given in Section \ref{EleAuxFKF}.)
One obtains
\begin{align}
P_{X}\Big(S_1+\cdots+S_n\le t<S_1+\cdots+S_{n+1}\Big)
\le \frac{(d_0t)^{n+1}}{(n+1)!}e^{-d_0t}.
\end{align}
\end{lemm}
\begin{proof}
To prove the lemma, let us consider a single electron on $\tilde{\Lambda}=[-\ell/2-1, \ell/2+1)^d \cap \BbbZ^d$.
We impose the periodic boundary conditions on this system as follows:
Let $\partial E_{\Lambda}$ be the set of pairs $\{x, y\}\in \tilde{\Lambda}\times \tilde{\Lambda}$ satisfying the following:
\begin{itemize}
\item there exists an $i\in \{1, \dots, d\}$ such that $x_i-y_i=\ell+1$;
\item for all $j\in \{1, \dots, d\}\backslash \{i\}$, $x_j-y_j=0$ holds. 
\end{itemize}
Let $\tilde{\Omega}=\tilde{\Lambda}\times \{-1, +1\}$.
Let $(t^{\rp}_{x y})$ be the hopping matrix with the periodic boundary conditions defined by
 $t_{xy}^{\rp}=t_{xy}$ if $\{x, y\}\in E_{\Lambda}$; $t^{\rp}_{xy}=t$ if $\{x, y\}\in \partial E_{\Lambda}$.
The corresponding kinetic energy of the electron is a self-adjoint operator  $h_0^{\rp}$ acting in $\ell^2(\tilde{\Omega})$ defined by 
\begin{align}
(h_0^{\rp} f)(x, \sigma)=\sum_{\sigma=\pm 1} \sum_{y\in \Lambda} t_{xy}^{\rp} (f(x, \sigma)-f(y, \sigma)),\ \ f\in \ell^2(\tilde{\Omega}).
\end{align}
By replacing (\ref{OneElDef}) with 
\begin{align}
P^{\rp}(Y_n=X| Y_{n-1}=Y)=\delta_{\sigma\tau}\frac{t_{xy}^{\rp}}{d_0},\ \ X, Y\in \tilde{\Omega},
\end{align}
we can construct a Feynman-Kac-It\^o formula for $h_0^{\rp}$.
Notice that because we consider the periodic boundary conditions, $d(x)$ is always constant: $d(x)=\sum_{
y\in \tilde{\Lambda}} t_{xy}=d_0$.
The probability measures $P$ and $P^{\rp}$ are related as 
\begin{align}
P(A)=P^{\rp}(A\cap \mathscr{R})\big/P^{\rp}(\mathscr{R}) ,
\end{align}
where $\mathscr{R}=\{{\bs m}\in M\, |\, X_{t}({\bs m}) \in \Omega\ \mbox{for all $0\le t$}\}$.
Accordingly, it suffices to prove that 
\begin{align}
P^{\rp}_{X}\Big(S_1+\cdots+S_n\le t<S_1+\cdots+S_{n+1}\Big)
=\frac{(d_0t)^{n+1}}{(n+1)!}e^{-d_0t}.
\end{align}
Indeed, 
taking   the definition of $S_n$ in Section \ref{FKIEl} into account, we have
\begin{align}
&P_{X}^{\rp}\Big(S_1+\cdots+S_n\le t<S_1+\cdots+S_{n+1}\Big)\no
=&\int_{\{t_1+\cdots+t_n\le d_0t <t_1+\cdots+t_{n+1}\}}
e^{-(t_1+\cdots+t_{n+1})} dt_1\cdots d t_{n+1}\no
=& \frac{(d_0t)^{n+1}}{(n+1)!}e^{-d_0t}.
\end{align}  
Thus, we are done. 
\end{proof}

\end{document}